\documentclass[journal,twoside,web]{ieeecolor}

\usepackage{generic}
\usepackage{amsmath,amssymb,amsfonts}
\usepackage{algorithmic}
\usepackage{graphicx}
\usepackage{textcomp}
\usepackage{comment}

\usepackage{amsthm}
\usepackage[utf8]{inputenc}
\usepackage[T1]{fontenc}
\usepackage{enumerate}
\usepackage[caption=false]{subfig}

\usepackage{soul}

\usepackage{hyperref}
\hypersetup{
    colorlinks=true,
    linkcolor=blue,
    filecolor=magenta,      
    urlcolor=cyan,
}

\usepackage[backend=biber,style=ieee,citestyle=ieee]{biblatex}
\AtBeginBibliography{\scriptsize}  
\newtheorem{theorem}{Theorem}
\newtheorem{proposition}{Proposition}
\newtheorem{corollary}{Corollary}
\newtheorem{lemma}{Lemma}
\theoremstyle{definition}
\newtheorem{definition}{Definition}%[section]

\newtheorem{remark}{Remark}

\usepackage{titlesec}
\titlespacing*{\subsubsection}{0pt}{\baselineskip}{0pt}
\def\BibTeX{{\rm B\kern-.05em{\sc i\kern-.025em b}\kern-.08em
    T\kern-.1667em\lower.7ex\hbox{E}\kern-.125emX}}
\begin{document}

\title{Barrier States Theory for Safety-Critical Multi-Objective Control}

\author{Hassan Almubarak, Nader Sadegh, and Evangelos A. Theodorou  
\thanks{This work was supported in part by the National Aeronautics and Space Administration (NASA) under ULI Grant 80NSSC22M0070.}
\thanks{H. Almubarak was with School of Electrical and Computer Engineering, Georgia Institute of Technology, Atlanta, GA, USA.
(e-mail: halmubarak@kfupm.edu.sa).}
\thanks{N. Sadegh is with the George W. Woodruff School of Mechanical Engineering, Georgia Institute of Technology, Atlanta, GA, USA. (e-mail: sadegh@gatech.edu).}
\thanks{E. A. Theodorou is with the Daniel Guggenheim School of Aerospace Engineering, Georgia Institute of Technology, Atlanta, GA, USA. (e-mail: evangelos.theodorou@gatech.edu).}}

\maketitle
\begin{abstract}
Multi-objective safety-critical control entails a diligent design to avoid possibly conflicting scenarios and ensure safety. This paper addresses multi-objective safety-critical control through a novel approach utilizing \textit{barrier states (BaS)} to integrate safety into control design. It introduces the concept of \textit{safety embedded systems}, where the safety condition is integrated with barrier functions to characterize a dynamical subsystem that is incorporated into the original model for control design. This approach reformulates the control problem to focus on designing a control law for an unconstrained system, ensuring that the barrier state remains bounded while achieving other performance objectives.

The paper demonstrates that designing a stabilizing controller for the safety embedded system guarantees the safe stabilization of the original safety-critical system, effectively mitigating conflicts between performance and safety constraints. This approach enables the use of various legacy control methods from the literature to develop safe control laws. Moreover, it explores how this method can be applied to enforce input constraints and extend traditional control techniques to incorporate safety considerations. Additionally, the paper introduces input-to-state safety (ISSf) through barrier states for analyzing robust safety under bounded input disturbances and develops the notion of input-to-state safe stability (IS$^3$) for analyzing and designing robustly safe stabilizing feedback controls. The proposed techniques and concepts are used in various examples including the design of proportional-integral-derivative-barrier (PIDB) control for adaptive cruise control.
\end{abstract}
\vspace{-3mm}
\vspace{-4mm}

\section{Introduction} \label{sec: Introduction}
\subsection{Background}
\IEEEPARstart{C}{ontrol} theory aims to develop mathematical laws or algorithms to derive a controlled dynamical system to a desired behavior while adhering to some performance objectives. 
Control systems often have different types of restrictions that need to be considered for a proper control. Those restrictions can be intrinsic or extrinsic to the system. In today’s fast growing, interdisciplinary technologies, control engineering has been a central element, from simple decision making problems to terrifically complex autonomous systems. Such systems often have diverse complicated restrictions or requirements that need to be simultaneously assured. An increasingly vital requirement in such systems is safety, in its various forms. Additionally, with the growing demand for autonomy across industries, this task remains daunting even in known environments. Therefore, there is a clear need for provably safe controls. Yet, the challenge of simultaneously satisfying performance and safety objectives, even when feasible, typically calls for trade-offs between them. 

Safety can be defined as the state of being secured or guarded against any type of danger. In other words, safety can be thought of as the avoidance of \textit{dangerous} states and the stay in harmless states. Lyapunov theory for ordinary differential equations (ODEs) has unquestionably influenced the introduction and the use of various mathematical notions to solve and analyse many problems in control theory. One of those notions is invariance, a mathematical property studied in different branches of mathematics, of subsets in the state space. For dynamical systems, a subset is (positively) invariant if at some point in time it contains the system's state then it contains the state for all (future) times \cite{blanchini1999set}. Evidently, given a dynamical system, safety can be represented and verified by the invariance of the set of permitted states in the state space and this set is referred to as the safe set. As the idea of sets' invariance is a crucial part of stability analysis in Lyapunov theory \cite{khalil2002nonlinear}, it can play a principal role in other control theory problems such as constrained or safe and robust control \cite{blanchini1999set}. Principally, the idea of invariant set of dynamical systems can be extended to \textit{control} systems in which a set is said to be \textit{controlled} invariant, also said to be viable, if for any initial condition in that set, the controlled system's trajectory is inside that set \cite{blanchini1999set}. This idea is of a huge interest as it can be used for control system's analyses of meeting different design requirements such as stability verification or constraints satisfaction as well as control synthesis.

As constraints are ubiquitous in control systems, different techniques have been proposed in the literature. Early in the control literature (before and at the start of the $21^{\text{st}}$ century), many of those techniques are based on rendering set invariance through Lyapunov analyses (see \cite{blanchini1999set,liu1994dynamical} and their references therein). Approaches that are considered in optimal control and optimization frameworks include using model predictive controls (MPC) to compute open-loop optimal controls \cite{mayne2000constrainedMPC,mayne2014mpc_recent} and using dynamic programming for feedback controls (see \cite{lin1991DDP_constrained_1,bertsekas2012dynamicprogramming} and the references therein). Another approach is reference governor which is an add-on paradigm to enforce state and/or control constraints pointwise in time through altering the reference input of the system \cite{gilbert2002nonlinear_reference_goven,kolmanovsky2014reference}. As it is the case for the unconstrained case, those approaches were developed mostly for linear systems or systems with linear (box) constraints. Nonetheless, with advancements in control theory and the increasing power of computational resources, these approaches have been extended, improved, and further investigated for nonlinear systems and general nonlinear constraint. Other approaches rely on barrier or barrier-like methods, which we focus on and build upon in this paper due to their flexibility, computational efficiency, and compatibility with existing control frameworks, including legacy controllers. We provide a detailed review of these methods in the control literature later in the paper as we develop the proposed theory.
%\subsection{Background and Brief Literature Review}
%\textcolor{green}{I'm considering not having background or review but telling a story and writing the background and review a long the paper especially in the early sections (maybe write a sentence about this in the intro). Or maybe brief background and then say that the background and review will be throughout the paper??}
\vspace{-7mm}
\subsection{Contributions and organization}
The original idea of barrier states was presented briefly and mainly for safe stabilization in a letter in \cite{Almubarak2021SafetyEC}. The idea was extended to discrete systems' trajectory optimization within dynamic programming \cite{almubarak2021safeddp}, to the min-max trajectory optimization case \cite{almubarak2021safeminmax} and various advanced control techniques \cite{aoun2023l1,cho2023mpcbas,oshin2024diffRobMPC}.
%to $\mathcal{L}_1$ adaptive control for constrained systems \cite{aoun2023l1}, to a model-predictive control (MPC) formulation \cite{cho2023mpcbas} as well as to a differentiable tube MPC \cite{oshin2024diffRobMPC}. 
This paper significantly extends our previous conference version in the following key directions: 
\begin{itemize}
    \item We refine and generalize the barrier-state (BaS) construction to be independent of the specific barrier function, drawing an analogy to dynamic compensation, and extend the formulation to general nonlinear systems beyond the control-affine class,
    \item provide a necessary and sufficient condition on the boundedness of the barrier state establishing safety guarantees,
    %\item derive a theorem that establishes the necessary and sufficient conditions for the existence of solutions to the newly formulated unconstrained (safety embedded) control problem,
    \item derive a theorem establishing necessary and sufficient conditions for solution existence to the newly formulated unconstrained (safety embedded) control problem,
%    \item We analyze the behavior of barrier-state feedback control under different controller designs, demonstrating its flexibility and influence on the closed-loop performance.
    \item investigate the design of BaS for input constraints,
    \item study safe robustness and robust safe stabilization through a novel proposition of input-to-state safety (ISSf) and input-to-state safe stability (IS$^3$) via barrier states for systems with input disturbance, and 
    \item demonstrate the versatility of the BaS framework through multiple control designs such as Lyapunov-based, LQR, and pole-placement across objectives including stabilization, robustness, and input constraints. In addition, we introduce a novel proportional-integral-derivative-barrier (PIDB) controller for adaptive cruise control, showcasing BaS integration with classical PID controls.
%    \item adopt well-known control techniques, such as proportional-integral-derivative (PID) controllers, with barrier states augmentation to develop an adaptive cruise controller termed the proportional-integral-derivative-barrier (PIDB) control, and
%    \item provide more control design examples with different control objectives, e.g. robustness and input constraints.
\end{itemize}

% A preliminary version of this work was presented in the conference publications [2] and [39]. The present paper adds to those two papers in the following important ways: the relations between the two forms of barrier functions are characterized; barriers with a higher relative degree are considered; the adaptive cruise control problem is extended from the lead vehicle’s speed being constant to the more realistic case of varying speed with bounded input force; and the lane keeping problem is considered under the proposed QP framework.

The paper is organized as follows. \autoref{sec: Problem Statement} introduces the problem statement with main set invariance and safe control definitions. \autoref{sec: Barrier States} is the main section which starts with reviewing barrier functions, their usage in dynamical systems and control theories, and the development of barrier and barrier-like based safe control methods. It then proceeds with discussing conventional conditions of barrier derivatives for safety verification and enforcement. The proposed ideas of barrier states and safety embedded systems as well as safe multi-objective control synthesis and conditions are subsequently presented. The idea is extended to input constraints afterwards. A tutorial example demonstrating the construction and behavior of the proposed safety embedded system, and its use in achieving almost-global safe stabilization, is subsequently presented. In \autoref{sec: input-to-state safe stability}, we develop the notion of input-to-state safety (ISSf) via barrier states and introduce the novel notion of input-to-state safe stability (IS$^3$) along with robust safe control and IS$^3$ control syntheses through safety embedded systems. \autoref{sec: Applications Examples} discusses various application examples of employing the proposed ideas in different problems including safety-critical linear control, with and without input constraints, adaptive cruise control through proportional-integral-derivative control with barrier states, and input-to-state safe stabilization for systems with input disturbance. Finally, concluding remarks are provided in \autoref{sec: Conclusion}. 
   
%\subsection{Notation}
\vspace{-1mm}
\section{Problem Statement} \label{sec: Problem Statement}
Consider the general class of nonlinear control systems
\begin{equation} \label{eq: control system dynamics}
    \dot{x}(t)= f(x(t),u(t))
\end{equation}
where $t \in \mathbb{R}_0^+, \ x \in \mathcal{X} \subset \mathbb{R}^n$, $u \in \mathcal{U} \subset \mathbb{R}^m$ and $f:\mathcal{X} \times \mathcal{U} \rightarrow \mathcal{X}$ is locally Lipschitz continuous with $x(0)=x_0$. %at which $f$ is forward complete. 
Unless necessary, time dependence is dropped throughout the paper for notational convenience. We consider the closed set $\overline{\mathcal{S}}\subset \mathcal{X}$ defined as the superlevel set of a continuously differentiable function $h(x): \mathcal{X} \rightarrow \mathbb{R}$, with $ \mathcal{S}$ and $\partial \mathcal{S}$ defined as its interior and boundary sets respectively such that
\begin{equation} \begin{split} \label{eq: safe set S}
    \overline{\mathcal{S}}:= \{ x \in \mathcal{X} \ | \ h(x) \geq 0\} \\
    \mathcal{S}:= \{ x \in \mathcal{X} \ | \ h(x) > 0\} \\
    \partial \mathcal{S}:= \{ x \in \mathcal{X} \ | \ h(x) = 0\}
\end{split} \end{equation} 
%The control system \eqref{eq: control system dynamics} must be controlled while ensuring that the system's states not leave the set $\mathcal{S}$.
We wish to design a \textit{safe feedback} control law $u(x)$ that realizes some predefined performance objectives while respecting the path constraint $h(x)>0$, and possibly control constraints which we discuss later in \autoref{sec: Barrier States}. A fundamental, yet obvious, assumption is that the multi-objective problem is feasible, i.e. there exists a solution to the problem.

Given the safety-critical control problem above, safety can be verified by the invariance of the permitted set. We now review \textit{controlled invariant sets} and safe controls to motivate our approach.
\begin{definition}[\cite{blanchini1999set}] \label{def: controlled invariant set}
The set $\mathcal{S} \subset \mathcal{X}$ is controlled invariant for the nonlinear control system $\dot{x}(t)= f(x(t),u(t))$ if for $x(0) \in \mathcal{S}$, there exists a continuous feedback controller $u=K(x)$, such that the closed-loop system $\dot{x}(t)= f\big(x(t),K(x(t))\big)$ %is forward complete with $x(t) \in \mathcal{S} \ \forall t \in \mathbb{R}^{+}$.
has the unique solution $x(t) \in \mathcal{S} \ \forall t \in \mathbb{R}^{+}$.
\end{definition}

\begin{definition} \label{def: safe control definition}
Given a set of possible initial conditions $\mathcal{X}_0 \subseteq \mathcal{S}$, the continuous feedback controller $u=K(x)$ is said to be safe if $\forall x(0) \in \mathcal{X}_0$, the safe set $\mathcal{S}$ is rendered forward invariant with respect to the resulting closed-loop system $\dot{x}(t)= f\big(x(t),K(x(t))\big)$. That is, by definition of the safe set $\mathcal{S}$ \eqref{eq: safe set S}, under $u=K(x)$ the solution $x(t)$ is such that 
\begin{equation} \label{eq: safety condition}
    \inf_{t \geq 0} h\big(x(t)\big) > 0 \ \forall t \in \mathbb{R}^{+} ; \ x(0) \in \mathcal{X}_0 
\end{equation}
\vspace{-2mm}
which is referred to as the safety condition.
\end{definition}
It is worth noting that in the above definitions, the closed-loop system is required to admit a unique solution for which local Lipschitz continuity of $f\big(x(t),K(x(t))\big)$ is a sufficient condition \cite{blanchini1999set}.

\section{Barrier States} \label{sec: Barrier States}

\subsection{Barrier Functions}
In the optimization literature, constrained optimization problems are commonly solved using penalty methods that transform the problem into an unconstrained one by adding a penalty function to the objective function. The penalty function is infinite when the constraint is violated and zero otherwise \cite{ben2020lecturesoptimization}, but this introduces discontinuities in the objective function. An alternative is to penalize the solution in the interior of the feasible set as it approaches the boundary using a penalty-like term added to the objective. This effectively forces the search for a solution to remain in the interior of the feasible set. These methods are known as barrier methods in which a smooth real valued mapping $\mathbf{B}: \mathbb{R} \rightarrow \mathbb{R}$ of the constraint function is used effectively preserving differentiability \cite{ben2020lecturesoptimization}. 

%In this work, we use $\mathbf{B}$ as the barrier operator, such as the popular inverse barrier function $\mathbf{B}(\eta) = \frac{1}{\eta}$, also known as the Carroll barrier, or the logarithmic barrier functions $\mathbf{B}(\eta) = -\log(\frac{\eta}{1+\eta})$. Hence, the proposed method in this paper is generic for any barrier function that satisfies the classical definition of barrier functions \cite{ben2020lecturesoptimization}. 
In this work, we use the mapping $\mathbf{B}$ as a barrier operator, defined using classical barrier functions such as the inverse (Carroll) barrier $\mathbf{B}(\eta) = \frac{1}{\eta}$ or the logarithmic barrier $\mathbf{B}(\eta) = -\log\left(\frac{\eta}{1+\eta}\right)$. We refer to $\mathbf{B}$ as an operator since it maps a scalar function to another scalar function, maintaining mathematical clarity and consistency throughout the manuscript. This formulation is generic and applicable to any function satisfying the classical definition of barrier functions in optimization \cite{ben2020lecturesoptimization}.
\begin{definition} \label{def: barrier functions}
A single valued function $\mathbf{B}: \mathbb{R} \rightarrow \mathbb{R}$ is a Barrier Function (BF), also known as an interior penalty function, if it is smooth on $(0,\infty)$ with a singularity at $0$ so that $\mathbf{B}(\eta)\rightarrow \infty$ as $\eta \rightarrow 0$ and $\mathbf{B}' \circ \mathbf{B}^{-1} (\eta)$ is $C^{\infty}$ everywhere.
\end{definition}
For a chosen barrier operator $\mathbf{B}$, we define the barrier $\beta: \mathcal{S} \rightarrow \mathbb{R}$ over the safety condition \eqref{eq: safety condition}, i.e. $ \beta(x):= \mathbf{B} \circ h\big(x(t)\big)$. This implies that $\beta(x) \rightarrow \infty$ as $x \rightarrow \partial \mathcal{S}$ and $\beta(x) < \infty \ \forall x \in \mathcal{S}$, ensuring satisfaction of the safety condition \eqref{eq: safety condition}.

\begin{proposition} \label{prop: boundedness of beta}
For the safety-critical control system \eqref{eq: control system dynamics}, given that $\ x(0) \in \mathcal{S}$, a continuous feedback controller $u=K(x)$ is safe, i.e. the safety condition \eqref{eq: safety condition} is satisfied and the safe set $\mathcal{S}$ is controlled invariant, if and only if $\beta(x(t))<\infty , \ \forall t \in \mathbb{R}^+_0$.
\end{proposition}
\begin{proof}
$\Rightarrow$ Assume that there exists a continuous control law $u=K(x)$ such that $\mathcal{S}$ is controlled invariant with respect to the closed-loop system $\dot{x}=f(x,K(x)), \ \ x(0) \in \mathcal{S}$ with the solution $x(t) \in \mathcal{S} , \ \forall t \in \mathbb{R}^+$. Then, by \autoref{def: controlled invariant set} and \autoref{def: safe control definition}, $h(x) >0 \  \forall t \in \mathbb{R}^+ $. Hence, by \autoref{def: barrier functions} and the properties of the barrier function $\beta$, $\beta(x(t))<\infty \ \forall t \in \mathbb{R}^+_0$.

$\Leftarrow$ Given that $h\big(x(0)\big) > 0 \Rightarrow \mathbf{B}(h(x_0))<\infty$, suppose that we have $ \beta(x(t))<\infty , \ \forall t \in \mathbb{R}^+$, under some continuous controller $u=K(x)$. By \autoref{def: barrier functions} and definition of the barrier function $\beta$, $h\big(x(t)\big) > 0 , \ \forall t \geq 0$. Thus, by \autoref{def: controlled invariant set}, $\mathcal{S}$ is controlled invariant with respect to the closed-loop system $\dot{x}=f(x,K(x))$ and by \autoref{def: safe control definition} $u=K(x)$ is safe.
\end{proof}

As with many optimization techniques, constrained methods such as penalty and barrier methods have naturally found application in the optimal control literature \cite{snyman1992PenaltyOptimalControl,fabien1996ExtendedPenaltyOptimalControl,xing1989exactPenaltyOptimalControl,abu2004nearly_const,hauser2006barrieroptimalcontrol_CDC2006,malisani2016interior}. In the context of nonlinear systems and control, inspired by Lyapunov theory, \textcite{prajna2003barrier} proposed Barrier Certificates (BC) for model validation of nonlinear systems with uncertain parameters. The proposed novel method validates the model by using BC as a means to examine any contradictions between model and data. The proposition is that safety can be verified through a transversality condition in which if there exists a barrier-like function that has a non-positive time derivative, then the model is invalid. This proposition is akin to that of \textcite{nagumo1942lage} concerning positive invariance that basically states that if the system's state is in the boundaries of some set with a tangent derivative or one that points inwards, then the system's trajectory remains in the set given that the system admits a globally unique solution \cite{blanchini1999set}. Barrier certificates were then used for safety verification of nonlinear hybrid systems in \cite{prajna2004safety} and for stochastic systems in \cite{prajna2004stochastic_safety_ver}. The proposed method is capable of determining the systems safety without explicit computation of reachable sets as opposed to the safety verification literature. Essentially, a BC is a continuously differentiable function that is positive in the unsafe set and negative in the safe set. For any initial condition that is in the safe set, the BC must have a negative definite time derivative certifying the safety of the system. % If a barrier certificate is found such that for any initial start from the safe set yields a negative definite time derivative of the barrier certificate, then the systems is safe. 
The negative definiteness of the barrier's time derivative was then relaxed in \cite{kong2013exponentialBarCert, dai2017barriercertrelax}. Nonetheless, these barrier certificates do not necessarily have barrier functions properties in the classical sense, i.e. the function's value approaches infinity if the system approaches the boundaries of the safe set, and are usually constructed as sum of squares polynomials. 

Concurrent with the work of \textcite{prajna2003barrier}, \textcite{Ngo2004backstepping_constrained} used barrier functions for state constraints to present an implicitly analogous condition but with an eye towards forming a Lyapunov function. The proposed technique drew inspiration from Lyapunov analyses and control and was used specifically for systems with velocity state box constraints within a backstepping control paradigm. \textcite{tee2009barrier_Lyap_automatica} further generalized the approach and the augmented Lyapunov function was termed a Barrier-Lyapunov function (BLF) whose roots can be traced to Lyapunov wells \cite{sane2001phdthesis}.

Later, inspired by Control Lyapunov Functions (CLFs) and BC, \textcite{wieland2007constructive} introduced Control Barrier Functions (CBFs) to propose a feedback method of enforcing safety in continuous time systems. In an attempt to develop safe stabilization, \textcite{ames2014control} and \textcite{romdlony2014uniting} proposed spiritually similar, albeit distinct, CLF-CBF unification techniques. The work in \cite{romdlony2014uniting} was further detailed with Lyapunov analyses in computing a unified CLF-CBF controller in \cite{romdlony2016stabilization} that in essence resembles the BLF approach \cite{tee2009barrier_Lyap_automatica} but for a more general class of constraints. \textcite{ames2014control} pioneered the CLF-CBF Quadratic Programs (QPs) paradigm which was further developed in \cite{ames2016CBF-forSaferyCritControl}. It is worth mentioning that in the early developments of CBFs, they were also referred to as barrier certificates as they are originally based on BCs and share the same analyses \cite{borrmann2015control,wang2016safetyBC,wang2018safe}. Additionally, it is pertinent to note that there is the \textit{zeroing} CBF (ZCBF) \cite{ames2016CBF-forSaferyCritControl}, which works with the safety function $h(x)$ instead of wrapping a barrier around $h(x)$, though it needs to satisfy certain conditions to be a valid ZCBF. In this work, we mainly use and discuss barrier functions as defined in \autoref{def: barrier functions} which we refer to as \textit{classical} barriers in some parts of the paper. 
\vspace{-3mm}
\subsection{Barrier Derivatives (Do we need to dictate the barrier's rate of change?)}
The CLF-CBF QP and the developed CBFs by \textcite{ames2016CBF-forSaferyCritControl} have attracted researchers attention to be adopted in various control frameworks and robotic applications \cite{borrmann2015control,agrawal2017discrete,wang2018safe,choi2020reinforcement} for multiple reasons. Besides the advantages of CBF safety filters within a QP formulation, an important question that \textcite{ames2016CBF-forSaferyCritControl} posed and proposed an answer to is about a less restrictive condition on the derivative of the barrier function. \textcite{ames2016CBF-forSaferyCritControl} proposed a practical relaxation to the strict negative definiteness of the classical barrier's time derivative, or equivalently the strict positive definiteness of the constraint's time derivative, similar to the relaxation of barrier certificates for model validation \cite{kong2013exponentialBarCert, dai2017barriercertrelax}. In other words, the barrier function is allowed to grow when the safety-critical states are far from the boundaries of $\mathcal{S}$ and its growth rate approaches to zero as the safety-critical states approach $\partial \mathcal{S}$. This effectively removes the requirement of having invariant sub-level sets to the set $\mathcal{S}$. Mathematically, the conventional condition on the barrier function $\dot{\beta} \leq 0 $ is relaxed to
\begin{equation} \label{eq: cbf relaxed condition}
    \dot{\beta} \leq \alpha \big( h(x) \big)
\end{equation}
where $\alpha$ is a class $\mathcal{K}$ function, which is a continuous function with a positive definite domain that is strictly increasing and is zero at zero, i.e. $\alpha(0)=0$. The earlier relaxation of barrier certificates in \cite{kong2013exponentialBarCert} is a special case of this condition, namely when $\alpha$ is selected to be a linear function of $h$.

From a set invariance perspective, this means that instead of requiring every sub-level set of $\mathcal{S}$ to be forward invariant, only a single sub-level set is invariant. This is brilliant. Nonetheless, a pitfall of the condition in \eqref{eq: cbf relaxed condition} is that one has to find the growth rate function $\alpha$. For control purposes, dictating the \textit{allowed} growth rate of the barrier function can result in infeasibility or a hard to satisfy condition. This poses the difficulty of finding a \textit{valid} CBF, i.e. a CBF with a feasible allowed growth rate such that a solution is attainable. 

Additionally, to get a \textit{nice} inequality condition that can be solved and used with QPs, control affinity is needed, which basically helps setting up a convex program. Furthermore, the QP approach is myopic which may result in further practical challenges when combined with other control techniques \cite{almubarak2021hjbcbf,krstic2023inverse}. A more troublesome condition is that we must have $u$ to show up in the condition. This requires a well-defined relative degree which may not exist and makes the problem hard even if it exists but in higher derivatives in which a high relative degree CBF needs to be found. Multiple methods have been developed to construct a higher-order CBF, if it exists. A common method is through Exponential Control Barrier Functions (ECBFs) \cite{nguyen2016exponential} which is similar to how high relative degree CLF and feedback linearization are dealt with. Another approach is to define a new invariant safe set as the intersection of several invariant sets \cite{xiao2019highcbf}. However, these approaches tend to unnecessarily restrict the allowable state space and are difficult to implement and tune. To avoid finding a high-order CBF altogether, some practical implementations instead apply an ad-hoc solution that modifies the CBF or the safety constraint to ensure a relative degree of $1$  \cite{pereira2021safe,long2021learning,xiao2020feasibility}. However, this is likely to affect the performance or safety of the resulting system with respect to the original objective and safety constraint.

We go back to the question in \cite{ames2016CBF-forSaferyCritControl}, ``what conditions should be imposed on $\dot{\beta}$ so that $\mathcal{S}$ is forward invariant?''. The answer of \textcite{ames2016CBF-forSaferyCritControl} replaces the conventional answer, $\dot{\beta} \leq 0$, with the condition in \eqref{eq: cbf relaxed condition}. In this paper, we propose another answer, which we show to have advantageous byproducts, and hope that it provides a different perspective to the problem and opens a new research direction in safety-critical and multi-objective control theory.

The answer we propose is, \textit{it is up to the control law!} In other words, the rate of change of the barrier function should be up to the control law as long as the barrier is bounded. In essence, we move the focus back to be on the boundedness of the barrier function, for which it was originally developed (\autoref{prop: boundedness of beta}). It is worth noting, however, that dictating the rate of change of the barrier might be desirable in some applications. Nonetheless, our proposed approach also allows for such a preference through the control design selection. Now a different question arises. How can we design a control law that achieves the performance objectives while ensuring boundedness of the barrier function? 

Our proposed idea is to augment the barrier dynamics, i.e. the barrier rate of change, to the open loop model of the safety-critical system. Accordingly, the multi-objective control problem is now to design a controller that regulates the augmented system. This idea is compelling and popular in the control literature, e.g. in integral control, adaptive control and estimations, in which a system-dependent state is augmented to the system to be controlled along with the system's state.  
\vspace{-3mm}
\subsection{Barrier States and Safety Embedded Systems} \label{subsec: Barrier States and Safety Embedded Systems}
The proposed method is general for multi-objective control, such as regulation, stabilization, tracking, trajectory optimization and planning, etc. To motivate our development, we start by studying safe stabilization with the proposed method as an objective since stability is a major objective in control theory by which many control problems can be recast including, for example, forcing the system to a specific configuration, tracking a reference signal, etc.  We show the use of the proposed approach in different applications thereafter.
\begin{definition} \label{def: safe stabilizability}
    The control system \eqref{eq: control system dynamics} is said to be safely stabilizable around the origin if there exists a \textit{safe} continuous feedback controller $u=K(x)$ such that the origin of the closed-loop system $\dot{x}=f(x,K(x))$ is asymptotically stable. That is, the closed-loop system is stable with $\inf_{t \geq 0} h\big(x(t)\big) > 0 , \ \forall t \in \mathbb{R}^{+}$ and 
    %there is $\delta$ such that $||x(0)||<\delta \Rightarrow \lim_{t %\rightarrow \infty} x(t) \rightarrow 0$. Hence, 
    there exists an open neighborhood $\mathcal{A}_{\text{safe}} \subseteq \mathcal{S}$ of the origin such that $\forall x(0) \in \mathcal{A}_{\text{safe}}$, the solution $x(t) \in \mathcal{S}, \ \forall t \in \mathbb{R}^+$, with $\lim_{t \rightarrow \infty} x(t) \rightarrow 0$. 
\end{definition}

\subsubsection{Barrier States in Multi-Objective Control}
For the safety-critical nonlinear system \eqref{eq: control system dynamics}, the barrier function, $\beta$, has the following rate of change
\begin{equation} \label{eq: barrier beta dynamics}
    \dot{\beta} = \mathbf{B}'\big(h(x)\big) L_{f(x,u)} h(x) 
\end{equation}
which can be used to describe the \textit{state of the barrier}. Note that the dynamics of $\beta$ depends on the system's dynamics. As the idea is to augment the system's dynamics with the barrier's dynamics, we need to at least ensure stabilizability of the augmented system. % Hence, we propose the \textit{barrier state (BaS)} equation  
%\begin{equation} \label{eq: barrier states dynamics (beta)}
%        \dot{\beta} = \mathbf{B}'\big(h(x)\big) L_{f(x,u)} h(x)  - \gamma \big( \beta - \mathbf{B} \circ h(x) \big)
%\end{equation}
%where $\gamma \in \mathbb{R}^+$.
To justify our proposition, let us write the barrier's dynamics as an output feedback problem, for which a dynamic compensation is used (see for example \cite[Chapter 12]{khalil2002nonlinear} for the analogy), % Letting $\phi(x):=\mathbf{B}'\circ \mathbf{B}^{-1} (\beta)$, we have
\begin{align}  \label{eq: barrier beta dynamics with its definition for output feedback}
\begin{split}
    & \dot{\beta} = \mathbf{B}'\circ \mathbf{B}^{-1}(\beta)  L_{f(x,u)} h(x) \\ 
    & y =\beta(x)
\end{split} 
\end{align} 
where $y=\beta(x)$ plays the role of our measured output. Then, the observer dynamics is given by
\begin{align} \label{eq: barrier states dynamics (beta_hat)}
    \dot{\hat{\beta}} = \mathbf{B}'\circ \mathbf{B}^{-1}(\hat{\beta}) L_{f(x,u)} h(x) - \gamma \big( \hat{\beta}-\beta(x) \big)
\end{align}
where $\gamma >0$ is analogous to the observer gain in the dynamic compensation problem and the negative definiteness ensures a converging reduction of the state estimation error.

Without loss of generality, let us assume that we are interested in stabilizing the system around the origin, i.e. $f(0,0) =0$. The barrier function $\beta$ might have an arbitrary value at the origin $(x=0)$ that depends on the constraint and the selected barrier given by $\beta_0:=\beta(0) = \mathbf{B} \circ h\big(0\big)$. %Therefore, we shift its equilibrium point to the origin by defining $z:=\beta-\beta_0$. Hence, the barrier state (BaS) equation becomes
To offset this, we introduce a new state variable $z$, termed the barrier state (BaS), whose dynamics vanish when $x=0$ and $u=0$, ensuring that the origin remains an equilibrium point. The dynamics of this BaS aligns with the observer equation \eqref{eq: barrier states dynamics (beta_hat)}, where $\hat{\beta}$ is replaced by $z + \beta(0)$.
\begin{align} \label{eq: barrier states dynamics (z)}
\begin{split}
    \dot{z} & =: f_b(x,z,u) \\
    & = \mathbf{B}'\circ \mathbf{B}^{-1}(z+\beta_0) L_{f(x,u)} h(x) - \gamma \big(z+\beta_0 - \beta(x)\big)
\end{split}
\end{align}
The following lemma shows that boundedness of the BaS $z$ guarantees satisfaction of the safety constraint, and importantly when $z(0)$ is set to $\beta\big(x(0)\big) - \beta_0$, we have $z(t) = \beta\big(x(t)\big) - \beta_0$.
\begin{lemma} \label{lemma: boundedness of z means boundedness of beta}
Let $\mathcal{X}_b \subset \mathcal{X}$, $\mathcal{U}_b \subset \mathcal{U}$ be bounded sets of \eqref{eq: control system dynamics} and let $\epsilon >0 $ be such that $z(0)\in (\beta(x(0))-\epsilon,\beta(x(0))+\epsilon)$. Then, there exists $\gamma>0$ such that the BaS $z(t)$ generated by the state equation in \eqref{eq: barrier states dynamics (z)} along the trajectories of safety-critical system \eqref{eq: control system dynamics} is bounded if and only if $\beta \big(x(t)\big)$ is bounded $\forall t \geq 0$.
\end{lemma}

\begin{proof}
$\Rightarrow$ Suppose that $z(t)$ is bounded for all $t\geq0$, with any $z(0) \in (\beta(x(0))-\epsilon,\beta(x(0))+\epsilon)$. Define $\tilde{z}=z+\beta_0-\beta(x)$ and let $\phi(x):=\mathbf{B}'\circ \mathbf{B}^{-1} (\beta)$, which is continuously differentiable. Then, 
\begin{equation} \label{eq: barrier state error dynamics (zt)}
\dot{\tilde{z}}  =  \Big( \phi \big(\tilde{z}+\beta(x)\big) - \phi\big(\beta(x)\big) \Big) L_{f(x,u)} h(x) - \gamma \tilde{z}
\end{equation}
It can be easily seen that the preceding differential equation has an equilibrium point at $\tilde{z}=0$. Moreover, if $
z(0)=\beta \big(x(0)\big)-\beta(0) \Leftrightarrow \tilde{z}(0)=0
$, then $\tilde{z}(t) =0 \ \forall t \geq 0$ and
\begin{equation*} z(t)=\beta(x(t))-\beta(0), \ \forall t \geq 0\end{equation*}
Consequently, $\beta(x(t))=z(t)+\beta(0)$ for all $t\geq0$ is bounded.

$\Leftarrow$ Suppose that $\beta(x(t))$ is bounded starting from some $x(0) \in \mathcal{S}$ and let $\epsilon>0$ be given. We shall show that $\exists \gamma >0$ such that $\tilde {z}(t) \rightarrow 0$ exponentially $\forall z(0)\in (\beta(x(0))-\epsilon,\beta(x(0))+\epsilon)$, hence establishing the boundedness of $z(t)=\beta(x(t))-\beta(0)+\tilde{z}(t)$.

We claim that if $|\tilde{z}(0)|<\epsilon$, then $\tilde{z}(t)<\epsilon$, $\forall t \geq 0$. We proceed to prove this by contradiction. If the claim is not true, by continuity of the solution, there exists $T>0$ such that $|\tilde{z}(T)|=\epsilon$ and $\sup_{t\in [0,T]}|\tilde{z}(t)|\leq \epsilon$. Then, by the Mean Value Theorem, the error dynamics in \eqref{eq: barrier state error dynamics (zt)} become
\begin{equation}
\dot{\tilde{z}}=\phi'(\zeta)L_{f(x,u)}h(x)\tilde{z}-\gamma\tilde{z}
\end{equation}
for some $\zeta \in (\beta(x)-\epsilon,\beta(x)+\epsilon)$ since $|\tilde{z} (t)|\leq \epsilon, \ \forall t\in [0,T]$. Moreover, the boundedness of $x(t)$, $u(t)$, and $\beta(x(t))$ imply
$$
    \gamma_0:=\sup_{x\in \mathcal{X}_b, u\in \mathcal{U}_b} \sup_{\zeta}|\phi'(\zeta)|  |L_{f ( x(t),u(t))} h(x)| <\infty
$$
and subsequently $|\phi (z+\beta_0) - \phi\big(\beta(x)\big)|\leq \gamma_0 |\tilde{z}|$.

Now setting $\gamma>\gamma_0$ and differentiating $\tilde{z}(t)^2$ with respect to time using \eqref{eq: barrier state error dynamics (zt)} and the preceding inequality, it follows that
$$
\frac{d}{dt} \tilde{z}(t)^2 \leq -2\underline{\gamma} \tilde{z}(t)^2, \quad \underline{\gamma}=\gamma-\gamma_0
$$
This implies that $|\tilde{z}|$ is non-increasing on $[0,T]$ contradicting that $|\tilde{z}(t)|$ reaches $\epsilon$ at $t=T$. Thus, $\sup_{t\geq 0}|\tilde{z}(t)|\leq \epsilon$. Furthermore, by Gronwall's inequality,  $\lim_{t\rightarrow \infty}\tilde{z}(t)=0$ exponentially, which completes the proof.
\end{proof}

\subsubsection{Multiple BaS vs a Single BaS} \label{subsec: multipl bas vs sinlge bas}
For multiple constraints, i.e. $h(x)$ is a vector of $q$ constraint functions such that $h(x) = [h_1(x), h_2(x), \dots, h_q(x)]^{\top}$, one can construct a barrier state for each constraint, interpreting the inequality in the safety condition \eqref{eq: safety condition} element-wise. Nonetheless, a single barrier state can be constructed to represent multiple constraints. 

Suppose we wish to aggregate $p$ of the $q$ constraints to form a single BaS. One can determine an aggregated safety function $\textbf{h}$ by working out
\begin{equation} \label{eq: aggregated h (safety function)}
\frac{1}{\mathbf{h}} = \sum_{i=1}^{p} \frac{1}{h_i}
\end{equation}
and then derive a single barrier state as described earlier in \eqref{eq: barrier states dynamics (z)}. An equivalent method, but perhaps one that involves a little more complicated derivation,  is to define a barrier function for each constraint, define the aggregated barrier function as the summation over all the barrier functions and then derive the barrier state, as we showed in \cite{Almubarak2021SafetyEC}.
%Let $\beta^p(x) =\sum_{i=1}^{p} \mathbf{B} \circ h_i (x)$ and $\beta^{p}_0 = \beta^p(0)$. Then,
%$$
%\dot{\beta}=\sum_{i=1}^{p} \mathbf{B}'(h_i) L_{f(x,u)} h_i(x)
%$$
%and therefore the BaS can be found to be
%\begin{align} \label{multi-BaS}
%\dot{z}=\sum_{i=1}^{p}  \Big( \mathbf{B}'(h_i) L_{f(x,u)} h_i(x) \Big) -\gamma \Big( z+\beta^{p}_0 - \sum_{i=1}^p \mathbf{B}\big( h_i(x) \big)  \Big)
%\end{align}
 
Creating a single BaS is often desirable, as it avoids introducing multiple nonlinear state variables into the safety embedded model, which could otherwise increase the complexity of control design. However, aggregating multiple constraints into a single BaS may introduce strong nonlinearities, potentially complicating control design.  Additionally, one loses flexibility in assigning different design parameters to individual constraints—for example, in weighting penalty terms in an optimal control formulation or tuning the conservativeness of each constraint independently. It is worth noting that in some cases, a single function $h(x)$ may naturally encode multiple constraints via a single safe set, in which case a single BaS is appropriate. This trade-off is problem-dependent and should be made at the discretion of the control designer. In \autoref{sec: Applications Examples}, we design a single BaS to represent multiple constraints for some applications and construct multiple ones in others to validate the proposed technique.

\subsubsection{Safety Embedded Systems and Controls}
Now, we are in a position to create the safety embedded model. By defining $z=[z_1, \dots, z_{n_b}]^{\top} \in \mathcal{B} \subset \mathbb{R}^{n_b}$, where $n_b$ is the number of constructed barrier states if more than one BaS is used and $\mathcal{B}$ is an open set of all possible values of $z$, and augmenting it to the system \eqref{eq: control system dynamics}, we get
\begin{equation} \begin{split}\label{eq:Augmented_Representation} 
    & \dot{x}= f(x,u) \\ 
    & \dot{z}= f_b(x,z,u)
\end{split}\end{equation}
Defining $\bar{x}=\begin{bmatrix} x \\ z \end{bmatrix} \in \bar{\mathcal{X}} \subset \mathcal{X} \times \mathcal{B}, \bar{f}=\begin{bmatrix} f \\ f_b\end{bmatrix}: \bar{\mathcal{X}} \times \mathcal{U} \rightarrow  \bar{\mathcal{X}}$, we can rewrite this systems as
\begin{equation} \begin{split} \label{eq: safety embedded system}
    & \dot{\bar{x}}= \bar{f}(\bar{x}, u)
 %   & \dot{\bar{x}}= \bar{f}(\bar{x})+ \bar{g}(\bar{x}) u \\ 
\end{split} \end{equation}
with $\bar{x}(0)=[x_0, \ z_0]^{\top}$ for $x_0 \in \mathcal{S}$ and $z_0=z(0)$. %and $z_0=\beta(x_0)-\beta(0)$.
This system is referred to as the \textit{safety embedded system}. Note that this idea resembles a dynamic state feedback control in which the goal is to stabilize the origin $(x=0,z=0)$ \cite[Chapter~12]{khalil2002nonlinear}. 

%\begin{lemma} \label{lemma: cont diff and fbar(0)=0}
%    The safety embedded dynamics in \eqref{eq: safety embedded system} is continuously differentiable with $\bar{f}(0)=0$.
%\end{lemma}
%\begin{proof}
%    By continuous differentiability of $f$, $h$ and $\mathbf{B}$, clearly $f_b$ is continuously differentiable as it is the product and addition of continuously differentiable functions. $\bar{f}(0) =[f(0) \ \ f_b (0)]^{\top}$ with $f(0) =0$ and $f_b(0) = \mathbf{B}'\big( h(0) \big) (0) - \gamma \big(0+\beta_0 - \mathbf{B} \circ h(0) \big) = 0 $.
%\end{proof}

\begin{remark} \label{remark: barrier function use in the control law instead of propagation}
    The safety embedded system, incorporating the BaS dynamics, enables the transformation of a constrained control problem into an unconstrained one. In the deterministic full-state feedback case considered here, the solution of the differential equation of the barrier state \eqref{eq: barrier states dynamics (z)} is simply the shifted barrier function, that is $z=\mathbf{B}\circ h(x) - \beta_0, \ \forall t\geq 0$ provided $z(0)=\beta\big(x(0)\big) - \beta_0$. Thus, in practice, $z$ can be computed directly from the current state $x$ eliminating the need to propagate \eqref{eq: barrier states dynamics (z)}. In fact, \autoref{lemma: boundedness of z means boundedness of beta} shows that the error is exponentially convergent even if the initial condition of the BaS does not match $\beta\big(x(0)\big) - \beta_0$.
\end{remark}

The safety objective is now \textit{embedded} within the dynamics of the closed-loop system.
Stabilizing the augmented system \eqref{eq: safety embedded system} enforces safe stabilization of the original safety-critical system \eqref{eq: control system dynamics}, rendering the safe set $\mathcal{S}$ controlled invariant.  %Before proving such results, it remains to show that there exists a solution for the new control problem, i.e. feedback stabilizability is not distorted by the augmentation. 
We now state the main result: solving the safety embedded control problem \eqref{eq: safety embedded system} is equivalent to solving the original constrained problem \eqref{eq: control system dynamics}–\eqref{eq: safety condition}.
\begin{theorem} \label{theorem: embedded system has a solution if and only if original system has one}
    The original control system \eqref{eq: control system dynamics} is safely stabilizable at the origin if and only if the safety embedded control system \eqref{eq: safety embedded system} is stabilizable at the origin for some $\gamma>0$. 
   %In other words, a solution to the \textit{unconstrained} control problem in \eqref{eq: safety embedded system} exists if there exists a solution to the  \textit{constrained} control problem in \eqref{eq: control system dynamics}-\eqref{eq: safety condition}.
\end{theorem}
\begin{remark}
    The theorem states that stabilizablitiy of the safety embedded system is necessary and sufficient for safe stabilizability of the original system. In other words, given a continuous feedback controller $u=K(\bar{x})$ such that the origin of the safety embedded closed-loop system, $\dot{\bar{x}}=\bar{f}\big(\bar{x}, K(\bar{x}) \big)$, is asymptotically stable, then there exists an open neighborhood $\mathcal{A}_{\text{safe}} \subseteq \mathcal{S}$ of the origin of the safety critical system \eqref{eq: control system dynamics} such that $u=K(\bar{x})$ is safely stabilizing (see \autoref{def: safe stabilizability}).
\end{remark}
\begin{proof}
  $\Rightarrow$ Suppose that the original control system is safely stabilizable in \eqref{eq: control system dynamics}-\eqref{eq: safety condition}. Then, by \autoref{def: safe stabilizability}, there exists a continuous feedback controller $u=K(x)$ such that given $x_0 \in \mathcal{A}_{\text{safe}}$, the closed-loop system $\dot{x} = f\big(x, K(x)\big)$ is stable and its solution $\big(x(t),x(0)=x_0 \in \mathcal{A}_{\text{safe}} \subseteq \mathcal{S} \big)$ defined for all $t \geq 0$ converges to the origin and satisfies $\inf_{t \geq 0} h\big(x(t)\big)>0$ for all $t \geq 0$, i.e. $x(t) \in \mathcal{S}, \ \forall t \in \mathbb{R}^+$. %By definition of asymptotic stability (cite Khalil's book), there exists an open neighborhood $\mathcal{A} \subset \mathcal{X}$ of the origin such that $x(t) \rightarrow 0$ as $t \rightarrow \infty$. Then, since $h\big(x(t)\big)>0 \ \forall t \geq 0$, $\mathcal{A} \cap \mathcal{S}=:\mathcal{A}_{\text{safe}}$.
  The properties of the original system, along with continuity of the solution $x(t)$ with respect to $t$ and the initial condition $x_0$ imply that the set of all possible trajectories of the states $\mathcal{X}_b$ with $x_0 \in \mathcal{A}_{\text{safe}}$ is bounded and so is the corresponding set of all input trajectories $\mathcal{U}_b=K(\mathcal{X}_b)$ \cite[Chapter~4]{khalil2002nonlinear}.
  
  Furthermore, by continuity of the barrier function over the state, $\beta:=[\beta_1 \; \cdots \beta_{n_b}]^{\top}$, the barrier image $\beta(x(t)) - \beta(0)$ is bounded and as $t \rightarrow \infty$, we have $\beta(x(t)) - \beta(0) \rightarrow 0$. This implies that for any  $\epsilon > 0$, and $z(0) \in (\beta(x(0)) - \epsilon, \beta(x(0)) + \epsilon)$, or equivalently $|\tilde{z}(0)| < \epsilon$, there exists $\gamma_i >0$, such that $\tilde{z}_i(t)=z_i(t)+\beta_i(0)-\beta(t)$ for all $i=1,2,...,n_b$, corresponding to each $\gamma_i$ and $z_i$. Consequently 
    %there exists an open neighborhood $\mathcal{Z} \subset \mathcal{B} \subset \mathbb{R}^{n_b}$ of the origin of the barrier's %state space in which 
    $z(t) \rightarrow 0$ as $t \rightarrow \infty$ by \autoref{lemma: boundedness of z means boundedness of beta}.
    
    Therefore, the augmented system is asymptotically stable since $x(t) \rightarrow 0$ and $z(t) \rightarrow 0$ as $t \rightarrow \infty$. Moreover, there exists an open neighborhood $\bar{\mathcal{A}}=\mathcal{A}_{\text{safe}} \times \mathcal{B}$ of the origin of the augmented system \eqref{eq: safety embedded system} in which the closed-loop system $\dot{\bar{x}} = \bar{f}\big(\bar{x}, K(x)\big)$ has the bounded solution $\big(\bar{x}(t), \bar{x}_0 \in \mathcal{S} \times \mathcal{B}\big)$ defined for all $t \geq 0$ and converges to the origin.

    $\Leftarrow$ Consider the continuous feedback controller $u=K(\bar{x})$ that renders the origin of the locally Lipschitz safety embedded system \eqref{eq: safety embedded system} asymptotically stable with a domain of attraction $\bar{\mathcal{A}}$. Then, there exist open neighborhoods $\mathcal{A}\subset \mathcal{X} \subset \mathbb{R}^n$ and $\mathcal{Z}\subset \mathcal{B} \subset \mathbb{R}^{n_b}$ of the origin with $\mathcal{Z}$ bounded such that $\mathcal{A} \times \mathcal{Z} \subset \bar{\mathcal{A}}$. For the vector $\beta:=[\beta_1 \; \cdots \beta_{n_b}]^{\top}$, by the continuity of $\tilde{\beta}(x)=\beta(x)-\beta(0)$ on $\{x \in \mathbb{R}^n: h(x)>0\}$, the inverse image of $\mathcal{Z}$ by $\tilde{\beta}$, $\tilde{\beta}^{-1} (\mathcal{Z})$, is an open neighborhood of the origin of the original system. Thus $\mathcal{A}_{\text{safe}}:= \mathcal{A} \cap \tilde{\beta}^{-1} (\mathcal{Z})$ is also an open neighborhood of the origin. 
    
    Hence, for any initial condition $x(0) \in \mathcal{A}_{\text{safe}}$, the trajectories $\bar{x}(t)$, $t\geq 0$, are bounded and converge to zero implying that $z(t)$ is bounded and converges to zero as well. Thus, by \autoref{lemma: boundedness of z means boundedness of beta}, $\beta(x)$ is bounded, guaranteeing the safety of $u=K(\bar{x})$ and the forward invariance of $\mathcal{S}$ with respect to the safety critical system \eqref{eq: control system dynamics} by \autoref{prop: boundedness of beta}.
\end{proof}

A key motivation for the barrier-state embedding framework is to integrate safety with performance objectives, e.g. stabilization, without explicitly relaxing either. When combined with a CLF, the resulting controller may resemble CLBF methods \cite{tee2009barrier_Lyap_automatica,romdlony2014uniting,romdlony2016stabilization}, a connection we leave for future study. For instance, our illustrative example at the end of this section uses a simple quadratic Lyapunov function on the embedded state, which resembles a CLBF. However, this arises naturally from the augmented dynamics and is not a direct combination of separate CLF and CBF terms. In general, a CLBF for the original constrained system or a CLF for the safety embedded system may be far more complex and not quadratic in the states.

%A motivation behind the development of barrier states embedding is integrating safety with other control objectives without \textit{direct relaxation} of any of the objectives. In essence, the safety control is coupled with the performance objectives, e.g. stability, to achieve the different objectives simultaneously, given that a solution exists to the original problem. Indeed, when set with a control Lyapunov functions approach, our method exhibits similar elements of the control Lyapunov barrier functions proposed in \cite{tee2009barrier_Lyap_automatica,romdlony2014uniting,romdlony2016stabilization} which we leave to investigate in a future work. It is worth noting, however, that the Lyapunov function of the safety embedded system, which is then a Lyapunov barrier function for the original system may not be a simple summation of the two functions as in the aforementioned work.

%\begin{remark}
% It is worth noting that one could work on controlling, e.g. stabilizing the BaS, for example by feedback linearization, Lyapunov functions construction, quadratic programs, etc., which will in fact result in the CBFs methods. However, our interest is in multi-objective control problems and thus the idea is to embrace different control techniques to work on the safety embedded system to achieve multi-objective control while avoiding conflicts between the different objectives. Note that performing feedback linearization, under the required conditions, on the barrier state alone results in equivalent results to those of CBF conditions.
%\end{remark}

\begin{remark}[Relative Degree] \label{remark: relative degree}
The proposed methodology is largely agnostic
%, which may be applied to a general class of nonlinear systems \eqref{eq: %control system dynamics}, is a state based control approach hence indifferent
to the relative degree associated with the safety function $h(x)$. As shown in \autoref{theorem: embedded system has a solution if and only if original system has one}, the theoretical guarantees hold regardless of relative degree. Thus, a control law for the safety embedded system \eqref{eq: safety embedded system} can be designed even when the relative degree is undefined, high, or mixed, as demonstrated in later examples. That said, some controller designs, such as Lyapunov-based or output feedback linearization, may still require knowledge of a well defined relative degree.

%It is crucial to note that the developed method is applicable to general nonlinear systems as described by \eqref{eq: control system dynamics} and is \textit{generally} agnostic to the relative degree of the safety function $h$, or equivalently of the barrier function, with respect to the system’s output. Specifically, the theory presented, particularly as outlined in \autoref{theorem: embedded system has a solution if and only if original system has one}, provides necessary and sufficient conditions for the problem irrespective of the relative degree of the barrier function. Consequently, a control law for the safety embedded system described in \eqref{eq: safety embedded system} can be designed even when the relative degree is not well-defined—whether it is high or the barrier exhibits a mixed relative degree, as illustrated in the subsequent sections on control synthesis and implementation examples. However, it is important to note that alternative control design methods, such as Lyapunov-based control or output feedback linearization techniques, may require knowledge of the relative degree.
\end{remark}

As previously discussed, designing feedback controllers based solely on barrier dynamics with fixed-rate constraints may conflict with other control objectives and lead to instability. In contrast, designing a controller for the safety embedded system \eqref{eq: safety embedded system} enables the use of standard control methods to simultaneously address safety and performance goals. The following sections show how even simple control techniques can yield effective safe controllers within this framework. In particular, we show how linear control, despite its simplicity, remains effective when applied to the linearized safety embedded system. This facilitates safe control design and is illustrated through several examples in \autoref{sec: Applications Examples}.

Without loss of generality, consider the linear system (e.g. linearization of \eqref{eq: control system dynamics} around the origin)
\begin{equation} \label{eq: linrarized system}
    \dot{x} = A x + B u
\end{equation}
%where $    A = \frac{\partial f}{\partial x}\Big|_{(x=0,u=0)}$ and $B = \frac{\partial f}{\partial u}\Big|_{(x=0,u=0)}$. 
The pair $(A,B)$ is assumed to be controllable or at least stabilizable. %Using the linearized dynamics, we can tackle different control problems. 
For a safety-critical system, to achieve the control objective safely, we deal with the safety embedded system \eqref{eq: safety embedded system}. Linearizing the safety embedded system around its origin, $(\bar{x}=0,u=0)$, yields 
\begin{equation} \label{eq: safety embedded linear system}
    \dot{\bar{x}} =\bar{A} \bar{x} + \bar{B} u
\end{equation}
where $\bar{A}=\begin{bmatrix} A & 0_{n \times 1} \\ \mathbf{B}' \circ \mathbf{B}^{-1}(\beta_0) \frac{\partial h}{\partial x}(0) A -\gamma \mathbf{B}'(\frac{\partial h}{\partial x}(0)) & -\gamma \end{bmatrix}$ and $\bar{B}=\begin{bmatrix} B \\ \mathbf{B}' \circ \mathbf{B}^{-1}(\beta_0) \frac{\partial h}{\partial x}(0)  B \end{bmatrix}$.
%\begin{gather} \label{eq: safety embedded linear system - A}
%    \bar{A}=\begin{bmatrix} A & 0_{n \times 1} \\ \mathbf{B}' \circ \mathbf{B}^{-1}(\beta_0) h_x(0) A -\gamma \mathbf{B}'(h_x(0)) & -\gamma \end{bmatrix}  \\ \label{eq: safety embedded linear system - B}
%    \bar{B}=\begin{bmatrix} B \\ \mathbf{B}' \circ \mathbf{B}^{-1}(\beta_0) h_x(0)  B \end{bmatrix}
%\end{gather}
Notice here that we assumed a single BaS for simplicity of presentation. 

\subsection{Input Constrained Safety Embedded Control} \label{subsec: input BaS}
Input constraints are inherent in control systems. Can we use the barrier states embedding method to enforce input constraints? Yes, we can. The idea is simple; define a new state for the input and define a new system whose input is the original input's derivative, a known trick in the literature used for different purposes. % This is a known trick in the literature used to \textcolor{magenta}{avoid some issues!? such as converting the system to a control-affine one!}

Consider the control system in \eqref{eq: control system dynamics} subject to some control constraints defined by $g(u)>0$. Notice that $g$ could be a vector of constraints, e.g. a set of box constraints for the vector $u$ which is the most popular form of input constraints. Let us define $v:= \dot{u}$ and $\tilde{x} = [x \ \ u]^{\top}$ which results in the system $\dot{\tilde{x}} =[
        f(x,u) \ \  v]^{\top}$.
%\begin{equation*}
%    \dot{\tilde{x}} = \begin{bmatrix}
%        f(x,u) \\ v
%    \end{bmatrix}
%\end{equation*}
%Defining a barrier function $\beta^u:=\mathbf{B}\big(g(u)\big)$ and a barrier state $z^u:= \beta^u - \beta^{u \circ}$, where $\beta^{u \circ}=\beta(u^{\text{eq}})$ and $u^{\text{eq}}$ is the desired input at the equilibrium point, yields
Defining a barrier function $\beta^u:=\mathbf{B}\big(g(u)\big)$ and a barrier state $z^u$, yields
\begin{equation} \label{eq: temporal barrier states dynamics z(t)}
     \dot{z}^{u} = f_b^u(u,z,v):= \mathbf{B}'\big( g(u) \big) L_{v} g(u) - \gamma \big(z^u+\beta^{u \circ} - \mathbf{B} \circ g(u) \big)
\end{equation}
where $\beta^{u \circ}=\beta(u^{\text{eq}})$ and $u^{\text{eq}}$ is the input at the equilibrium point. Let $\bar{x} = [\tilde{x} \ 
 \ z^u]^{\top}$. Then, the safety embedded system can be derived to be
\begin{equation*}
    \dot{\bar{x}} = [f(x,u) \quad  v \quad f_b^u(u,v)]^\top
\end{equation*}
which can also be written as 
\begin{equation*}
    \dot{\bar{x}} = \begin{bmatrix}   f(x,u) \\ 0 \\ - \gamma \big(z^u+\beta^{u \circ} - \mathbf{B} \circ g(u) \big)   \end{bmatrix} +  \begin{bmatrix}   0 \\ 1 \\ \mathbf{B}'\big( g(u) \big) \frac{\partial g}{\partial u}   \end{bmatrix} v
\end{equation*}

Note that such formulation might affect the relative degree of the system with respect to some constraints but that may not need to be a concern as discussed in \autoref{remark: relative degree} and as we demonstrate later in an implementation example.

\subsubsection*{Illustrative Example: Almost-Global Safe Stabilization}
To illustrate the core ideas of our framework, we present a simple example that highlights the construction and behavior of safety embedded dynamics using BaS. This setup also serves as the basis for the multi-agent system in \autoref{sec: Applications Examples}.

Consider a velocity-controlled point-mass robot navigating toward the origin while avoiding an obstacle centered at $(2,0)$ with radius $0.5$. The system dynamics are given by
$\dot{x} = \begin{bmatrix}  \dot{\rm{x}} \\ \dot{\rm{y}} \end{bmatrix} = u$, and the safety constraint is defined by $\mathcal{S}=\{x \in \mathbb{R}^2, \ h(x) =  (\rm{x} - 2)^2 + \rm{y}^2 - 0.5^2>0\} $. Using the inverse barrier function $\beta(x)=\frac{1}{h(x)}$, the barrier state dynamics are governed by
\begin{equation*}
    \dot{z} = -(z+\beta_0)^2 [2({\rm{x}}-2), \ 2{\rm{y}}] u - \gamma (z+\beta_0 -\beta(x))
\end{equation*}
 and $\beta_0=\beta(0)=\frac{1}{3.75}$. Then, the safety embedded system is  
\begin{equation*}
        \dot{\bar{x}} = \begin{bmatrix}
        u \\ g_z^\top u - \gamma(z+\beta_0 - \beta(x)) 
    \end{bmatrix}
\end{equation*}
where $g_z = -(z+\beta_0)^2 [2({\rm{x}}-2), \ 2{\rm{y}}]^\top$. While the controller can be chosen using any unconstrained control method, the closed-loop performance, such as the size of the region of attraction, depends on the control design. Here, we adopt a Lyapunov-based approach using Sontag's formula \cite{Sontag2009feedbackStablization}. Let $V = \frac{1}{2}\bar{x}^{\top} P \bar{x}$, with $P=\text{diag}(1,1,\kappa)$ and $\kappa>0$, then the feedback control law is
\begin{equation*}
    u = - x - \kappa g_z^\top z = \begin{bmatrix} - {\rm{x}} + 2 \kappa z ({\rm{x}} - 2)(z + \frac{1}{3.75})^2 \\ - {\rm{y}} + 2 \kappa  z {\rm{y}} (z + \frac{1}{3.75})^2 \end{bmatrix}
\end{equation*}
which asymptotically stabilizes the origin of the safety embedded system. The gain $\kappa$ governs how aggressively the robot steers away from the obstacle. Consequently, given asymptotically stability and boundedness of $z$, the origin of the original system is asymptotically \textit{safely} stable. Indeed, setting $z(0)=\beta \big(x(0)\big)-\beta_0$, it can be easily seen that $
\dot{V} = -\bar{x}^\top \begin{bmatrix} I_{2\times2} &  \kappa g_z \\  \kappa g_z^{\top} & \kappa^2 g_z^{\top}g_z \end{bmatrix}\bar{x}
$, which is semi-negative definite everywhere. However, due to the inherent topological obstruction, the closed-loop system admits another saddle equilibrium along the invariant half-line $\mathcal{L}_{0}=\{(\rm{x}, 0): x>2.5\}$ reflecting a discrete decision point as described by \textcite{sontag1999stability}. As a result, the closed-loop system achieves almost-global safe stabilization on the feasible region excluding $\mathcal{L}_{0}$, i.e. $\mathcal{A}_{\text{safe}} = \{ x\in \mathbb{R}^2, h(x)>0 \} \setminus \mathcal{L}_{0}$. \autoref{fig: single integ CLF almost-global} illustrates this behavior achieved by the barrier-states feedback controller with $\kappa=1$. The vector field shows how trajectories bend smoothly around the obstacle, guided by the repulsive effect of the barrier-state feedback. All safe initial conditions converge to the origin, except those lying exactly on the invariant half-line beyond the obstacle, which converge to the unstable saddle equilibrium.

\begin{figure}[t]
        \centering
    \includegraphics[trim=40 15 40 35, clip, width=\linewidth]{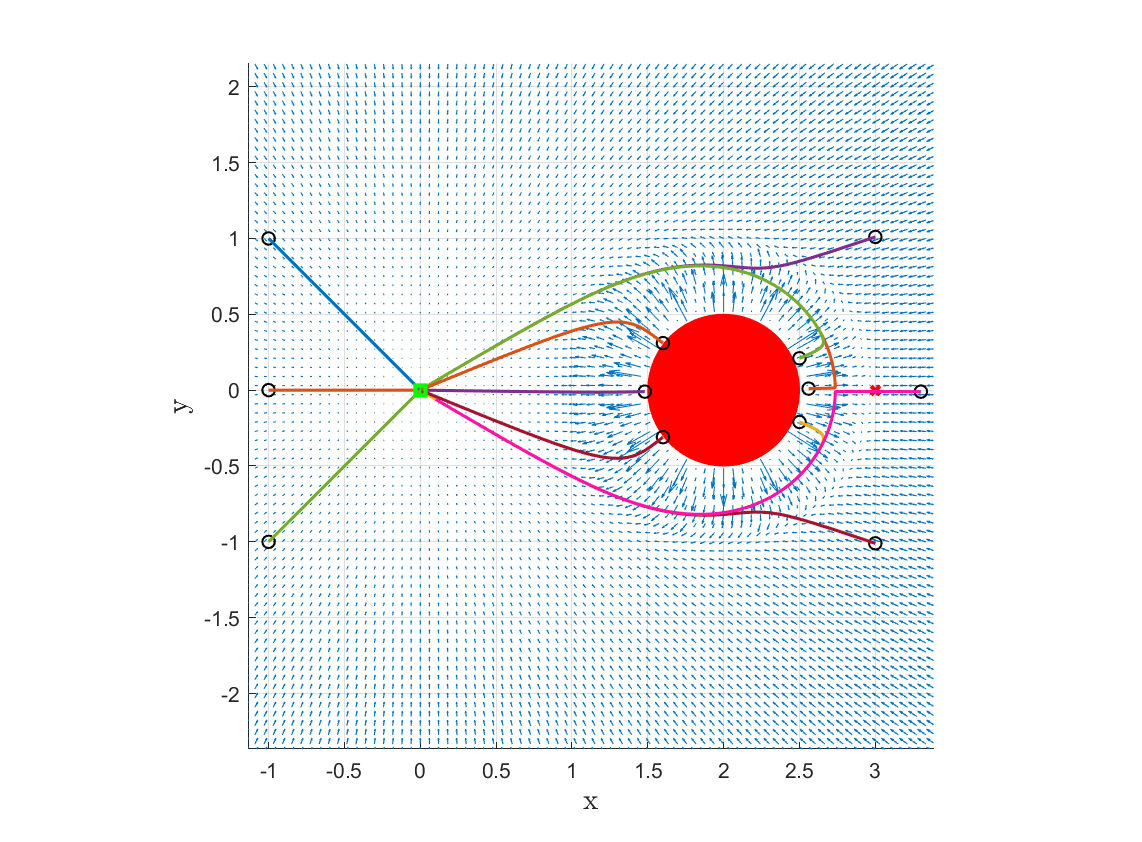}
      \caption{Simulations and phase portrait of the closed loop system under the BaS feedback controller $u=- x - g_z z$. The green $\textcolor{green}{\square}$ is the almost-globally stable equilibrium and the red $\textcolor{red}{\times}$ is the unstable one.}
      \label{fig: single integ CLF almost-global}
      \vspace{-4mm}
\end{figure}

%\begin{remark}
%    \textcolor{magenta}{One can derive safe stabilization guarantees for the input constrained problem similar to those presented in \autoref{theorem: embedded system has a solution if and only if original system has one} and \autoref{theorem: embedded system has a solution if and only if original system has one} which we omit. Additionally, such formulation might affect the relative degree of the system with respect to some constraints but that may not need to be a concern as we demonstrate later in our implementation examples.}
%\end{remark}

%\begin{remark}
%Although not common, one could use the same trick to deal with state-input inter-constraints, i.e. constraints which are functions of both inputs and states.
%\end{remark}

%\subsubsection{Example here!?}

\section{Input-to-State Safety (ISSf) and Input-to-State Safe Stability (IS\texorpdfstring{$^3$}{3})} \label{sec: input-to-state safe stability}
Safety-critical control design has achieved a lucid success in the past decade. Barrier methods discussed earlier have been a major element in the recent control literature. Nonetheless, a robust control design is crucial in ensuring effective control design, especially in safety-critical control. \textcite{sontag1989smooth} pioneered the notion of input-to-state stability (ISS) as a means of quantifying systems’ robustness to disturbances and constructing robustly stabilizing controllers. Motivated by the work of Sontag and the proposed feedback control redesign with the gradient of Lyapunov functions, \textcite{romdlony2016newnotionISSf} developed the notion of input-to-state safety (ISSf) following the footsteps of their earlier work of combining Lyapunov and barrier functions in \cite{romdlony2016stabilization}. Subsequently, \textcite{kolathaya2018ISSfCBF} formulated a related local ISSf notion within the CBF-QP framework \cite{ames2016CBF-forSaferyCritControl}, targeting controlled invariance of the safe set. Recently, \textcite{krstic2023inverse} explored ISSf properties in the context of inverse optimality of safety filters, using a CBF-QP framework interpreted through a zero-sum differential game, minimax optimal control, formulation. Notably, although ISSf is a barrier-based notion, similar ideas such as non-overshooting control were pursued earlier in the backstepping literature by \textcite{Krstic2006NonovershootingControl}.

In this section, we extend the barrier-state framework to analyze robustness under bounded external inputs or disturbances. Inspired by \textcite{romdlony2019robustness}, we define a notion of input-to-state safety based on barrier states and safety embedded systems. We then formalize the notion of input-to-state safe stability (IS$^3$) as a robustness property that unifies safety and stability under the proposed framework.

\vspace{-4mm}
\subsection{Input-to-State Safety via Barrier States}
%Consider the safety critical control \eqref{eq: control system dynamics} in which a disturbance $d$ enters the system through the control channels that is assumed to be locally essentially bounded, i.e. $d \in \ell^{\infty}$ where $\ell^{\infty}$ is the vector space of bounded sequences in $\mathbb{R}^m$. 
% As we shall see, the performance of the (controlled) system in handling the disturbances and staying in the safe region depends on the disturbance bounds. Based on that, we provide the following definition of robustly controlled invariant sets, a slight modification of that provided by \textcite{blanchini1999set}.
Consider the safety-critical system \eqref{eq: control system dynamics} subject to a disturbance $d \in \ell^{\infty}$, i.e., a locally essentially bounded input entering through the control channel. $\ell^{\infty}$ is the vector space of bounded sequences in $\mathbb{R}^m$. 

The system's ability to remain within the safe set under disturbance depends on the disturbance bounds. We therefore introduce a modified notion of robust controlled invariance, adapted from \textcite{blanchini1999set}. 
% That is, we are interested in rendering the safe set \textit{robustly} controlled invariant.
\begin{definition} \label{def: robustly controlled invariant set}
The set $\mathcal{S} \subset \mathcal{X}$ is robustly controlled invariant for the nonlinear control system $\dot{x}(t)= f(x(t),u(t), d(t))$ if for $x(0) \in \mathcal{S}$ and all $d(t) \in \mathcal{D}$, there exists a continuous feedback controller $u\big(x(t)\big)=K\big(x(t)\big)$, such that the closed-loop system $\dot{x}(t)= f\big(x(t),K(x(t)),d(t)\big)$ %is forward complete with $x(t) \in \mathcal{S} \ \forall t \in \mathbb{R}^{+}$.
has the unique solution $x(t) \in \mathcal{S} \ \forall t \in \mathbb{R}^{+}$.
\end{definition}

\begin{remark} \label{remark: ISSf feedback control design}
    The purpose of the ISSf concept, as well as the original ISS notion, is to assess the performance of a designed feedback controller in the presence of bounded disturbances. For the definition provided, we assume that $\mathcal{D} \subset \ell^{\infty}$ and $d$ influences the system through $u$. Consistent with the convention and existing literature \cite{sontag1989smooth,khalil2002nonlinear,mironchenko2023ISS-book}, $u$ is treated as the external input or disturbance. This perspective is based on viewing the system as a closed-loop configuration with some feedback controller, which will be formulated more explicitly in the control design discussion. It is also assumed that the control system can be made ISS \cite{sontag1989smooth} and ISSf \cite{romdlony2016newnotionISSf}.
    %The purpose of the ISSf notion, and originally the ISS notion, is to evaluate a designed feedback controller's performance under the effect of bounded disturbances. Therefore, for the preceding definition, we consider the case where $\mathcal{D} \subset \ell^{\infty}$ and $d$ enters the system through $u$. In what follows, following the convention and the literature \cite{sontag1989smooth,khalil2002nonlinear,mironchenko2023ISS-book}, $u$ is considered as the external input or disturbance since one can look at the system as a closed-loop system of some feedback controller, which we will formulate more explicitly when discussing the control design. Note that it is assumed here that the control system can be rendered ISS \cite{sontag1989smooth} and ISSf \cite{romdlony2016newnotionISSf}.     
    %\textcolor{magenta}{It is worth noting that for the controls part, we rely on the fact that if there exists a feedback control law such that the origin of the system is globally asymptotically stable (GAS), then there exists a \textit{modified} feedback control law such that the disturbed system is ISS \cite{sontag1989smooth}. For the local case, similar arguments can be achieved, see e.g. \cite[Chapter~2]{mironchenko2023ISS-book}.}
\end{remark}

%\textcolor{magenta}{\textbf{Note:} This is (and the notion of ISS generally) is connected to the time-dependent analysis in some sense since the disturbance makes our system non-autonomous and is basically using uniform stability arguments.}

For completeness, we review the definitions of class $\mathcal{K}$ functions, class $\mathcal{KL}$ functions and ISS systems \cite{sontag1989smooth,khalil2002nonlinear,mironchenko2023ISS-book}.

\begin{definition} \label{def: class K}
    A function $\alpha:[0,a) \rightarrow [0,\infty)$  is a class $\mathcal{K}$ function if it is continuous, strictly increasing and $\alpha(0)=0$. It is a class $\mathcal{K}_{\infty}$ if $a = \infty$ and for some $r \in \mathbb{R}$, as $ r \rightarrow \infty$, $\alpha (r) \rightarrow \infty$. 
\end{definition}

\begin{definition} \label{def: class KL}
    A function $\tilde{\alpha}:[0,a)\times[0, \infty) \rightarrow [0,\infty)$  is a class $\mathcal{KL}$ function if it is continuous, for each fixed $s \in  \mathbb{R}$, and some $r \in \mathbb{R}$, we have $\tilde{\alpha}(r,s)$ being a class $\mathcal{K}$ function w.r.t. $r$, and for each fixed $r$, we have $\tilde{\alpha}(r,s)$ to be decreasing w.r.t. $s$ such that $\tilde{\alpha}(r,s)\rightarrow 0$ as $s \rightarrow \infty$.
\end{definition}

\begin{definition} \label{def: ISS}
    The system \eqref{eq: control system dynamics} is input-to-state stable (ISS) if there exist $\tilde{\alpha} \in \mathcal{KL}$ and $\alpha_u \in \mathcal{K}$ such that $\forall x_0 \in \mathcal{X}, \ \forall u \in \ell^{\infty}$ and $t\geq 0$ we have
    \begin{equation} \label{eq: input-to-state stability condition}
        ||x(t)|| \leq \tilde{\alpha}(||x_0||, t) + \alpha_u (||u||_{\infty})
    \end{equation}
    It is locally input-to-state stable (LISS), if the condition \eqref{eq: input-to-state stability condition} only holds for $\forall x_0 \in \mathcal{X}_r, \ ||u||_{\infty}\leq r_u$ and $t \geq 0$, where $\mathcal{X}_r:=\{ x \in \mathcal{X}: ||x|| \leq r_x$\}.
\end{definition}

It is worth noting that ISS and LISS are closely related to uniform stability and uniform global asymptotic stability of dynamical systems \cite[Chapter~4]{khalil2002nonlinear}. Leveraging the formulation of barrier \textit{states},we construct input-to-state safety (ISSf) arguments analogous to those in the ISS framework. In what follows, we use the shifted barrier state $z$ whose dynamics vanish at the origin of the augmented system, allowing for a direct adaptation of standard ISS arguments. Notably, the adapted arguments are relaxed in that we require only uniform stability of the barrier state, rather than uniform asymptotic stability, under vanishing input ($u \rightarrow 0$ as $t \rightarrow \infty$). 

\begin{definition} \label{def: ISSf}
The safety embedded system is \emph{input-to-state safe (ISSf)} if there exist functions $\alpha_{z}, \alpha_{u} \in \mathcal{K}$ such that for all $x_0 \in \mathcal{X}_0$, $||z_0|| \leq \alpha_{z}(||x_0||)$, and all bounded inputs $u \in \ell^{\infty}$, the following condition holds for all $t\geq 0$:
            \begin{equation} \label{eq: input-to-state safety condition}
            ||z(t)|| \leq \alpha_{z}(||x_0||) + \alpha_{u}(||u||_{\infty}).
            \end{equation}
It is \emph{locally ISSf}(LISSf) if the above condition holds for all initial conditions $x_0$ within a subset $\mathcal{X}_{r_x} \subseteq \mathcal{X}_0$ and $||z(0)|| \leq Z_r < \infty$, with $||u||_{\infty}$ bounded, ensuring the existence of solutions for all $t\geq 0$, where $\mathcal{X}_{r_x}:=\{ x \in \mathcal{X}_0: ||x|| \leq r_x = \arg \max z(t) \in \mathcal{B}$\} and $Z_r:= \mathbf{B} \circ h(r_x)$.
\end{definition}
\begin{remark}
    The proposed condition \eqref{eq: input-to-state safety condition} is akin to that in \cite{romdlony2019robustness} albeit the latter focuses on staying away from the unsafe region. The barrier-state framework facilitates direct use of standard stability arguments from dynamical systems. In essence, in the absence of disturbances, $z(t)$ is uniformly stable implying that the system stays within the safe set. However, if the disturbance is large enough to violate \eqref{eq: input-to-state safety condition}, safety may no longer be guaranteed. For instance, if the initial condition lies near the boundary of the unsafe set, i.e. $z_0$ is large, even a small disturbance can cause $z(t)$ to become unbounded, violating \eqref{eq: input-to-state safety condition}.
\end{remark}

\subsection{Input-to-State Safe Stability (IS\texorpdfstring{$^3$}{3})}
While condition \eqref{eq: input-to-state safety condition} enables assessment of ISSf, a key objective is to evaluate both ISSf and ISS together, as this can inform the design of robust, safely stabilizing control laws. Building on the safety embedded system framework in \autoref{subsec: Barrier States and Safety Embedded Systems}, we now introduce the notion of input-to-state safe stability (IS$^3$).

\begin{definition}
    The safety critical system \eqref{eq: control system dynamics} is said to be input-to-state safely stable (IS$^3$) if it is ISS (\autoref{def: ISS}) and ISSf (\autoref{def: ISSf}) with $\alpha_z \in \mathcal{KL}$.
    It is locally input-to-state safely stable (LIS$^3$) if it satisfies both properties locally. %or is ISS and LISSf, LISS and ISSf or LISS and LISSf. 
\end{definition}

\begin{remark}
%    For the safe stabilization problem, notice that the state space is now considered to be $\mathcal{S}$ and as mentioned earlier, the set of initial states must be in $\mathcal{S}$, as indicated in \cite{romdlony2019robustness} for the stabilization problem with guaranteed safety. Hence, the new notion of input-to-state safe stability (IS$^3$) is defined over $\mathcal{S} \subset \mathcal{X}$.  
In the safe stabilization problem, the state space is restricted to the safe set $\mathcal{S}$, and thus all initial conditions must lie within $\mathcal{S}$. Hence, the new notion of input-to-state safe stability (IS$^3$) is explicitly defined over $\mathcal{S} \subset \mathcal{X}$, consistent with existing literature \cite{romdlony2019robustness}.
\end{remark}

Based on those definitions, we can establish the following significant proposition that allows us to work with the safety embedded system to achieve IS$^3$ for the safety-critical system. %It is worth noting that for the safety embedded control problem, we are interested in studying the IS$^3$ property of the closed-loop system $\dot{\bar{x}}=\bar{f}(\bar{x},K(\bar{x})+u)$, which results in an a feedback interconnected system. In such a case, under regularity conditions, the following lemma for interconnected systems is important before introducing our proposition. 
%\begin{lemma}[{\cite[Lemma~3.2]{mironchenko2023ISS-book}}] \label{lemma: interconnected ISS}
%    A forward complete subsystem $i$, is ISS if and only if $\exists \alpha_{ij}, \alpha_{ui} \in \mathcal{K}$ and $\exists \tilde{\alpha}_i \in \mathcal{KL}$ such that for all $ x_{i0} \in X_i$, all internal inputs $x_{\neq i} \in \ell^{\infty}$, for $i=1,\dots,N$ where $N$ is the dimension of the feedback interconnected system, and all external inputs $u$, we have 
%    \begin{equation}
%        ||x_i(t)|| \leq \tilde{\alpha}_i(||x_{0i}||,t) + \sum_{j\neq i}\alpha_{ij}(||x_{0j}||) + \alpha_{ui}(||u||_{\infty})
%    \end{equation}
%\end{lemma}

\begin{proposition} \label{prop: safety system is ISS iff orignal is IS3}
    The safety embedded system is (locally) ISS \eqref{eq: safety embedded system} if and only if the original system $\eqref{eq: control system dynamics}$ is (locally) IS$^3$.
\end{proposition}
\begin{proof}
    We prove the general case and the local case evidently follows.
    
    $\Rightarrow$ Assume that the original system is IS$^3$. Then, $\exists \tilde{\alpha}_1, \ \alpha_z \in \mathcal{KL}$ and $\alpha_{u1}, \ \alpha_{u2} \in \mathcal{K}$ s.t. $\forall x_0 \in \mathcal{S}, \ \forall u \in \ell^{\infty}$ and $t\geq0$,
    \begin{align*}
                & ||x(t)|| \leq \tilde{\alpha}_1(||x_0||, t) + \alpha_{u1}(||u||_{\infty}) \\ 
                & ||z(t)|| \leq \alpha_z(||z_0||,t) + \alpha_{u2}(||u||_{\infty})
    \end{align*}
    Squaring both sides of both inequalities and adding them yields,
    \begin{align*}
     & ||x||^2 + ||z||^2 \leq \big(\tilde{\alpha}_1 + \alpha_{u1} \big)^2 +  \big(\alpha_z + \alpha_{u2} \big)^2 \\
                 \Rightarrow & ||\bar{x}||^2 \leq \big( \tilde{\alpha}_1 + \alpha_{u1} + \alpha_z + \alpha_{u2} \big)^2 \\
                 \Rightarrow & ||\bar{x}|| \leq \tilde{\alpha}_1 + \alpha_{u1} + \alpha_z + \alpha_{u2} 
    \end{align*}
    Define $\alpha_u(||u||_{\infty}) = \alpha_{u1}(||u||_{\infty}) + \alpha_{u2}(||u||_{\infty})$ and $\tilde{\alpha} (||\bar{x}_0||,t) = \tilde{\alpha}_1(||\bar{x}_0||,t) + \alpha_z(||\bar{x}_0||,t)$. By definition of $\bar{x}$, %and noting that $\alpha_z(||z_0||)$ is a constant, 
    %$||\bar{x}_0||^2 = ||x_0||^2 + ||z_0||^2$ and, consequently,
    $||x_0||\leq ||\bar{x}_0||$ and $||z_0||\leq ||\bar{x}_0||$,
    and hence
    \begin{align*}
            \tilde{\alpha}_1(||x_0||,t) + \alpha_z(||z_0||,t) & \leq \tilde{\alpha}_1(||\bar{x}_0||,t) + \alpha_z(||\bar{x}_0||,t)  \\
            & =\tilde{\alpha} (||\bar{x}_0||,t)
    \end{align*}
    Therefore, $||\bar{x}(t)|| \leq \tilde{\alpha}(||\bar{x}_0||, t) + \alpha_u(||u||_{\infty})$, 
    %\begin{equation*}
    %        ||\bar{x}(t)|| \leq \tilde{\alpha}(||\bar{x}_0||, t) + \alpha_u(||u||_{\infty})
    %\end{equation*}
    $\forall \bar{x}_0 \in \mathcal{S} \times \mathcal{B}, \ \forall u \in \ell^{\infty}$ and $t\geq0$.

    $\Leftarrow$ Suppose that the safety embedded system is ISS. Then, by definition we have that $\exists \tilde{\alpha} \in \mathcal{KL}$ and $\exists \alpha \in \mathcal{K}$ such that $\forall \bar{x}_0 \in \mathcal{S}\times\mathcal{B}, \ \forall u \in \ell^{\infty}$ and $t\geq 0$ we have
    \begin{equation*} 
        ||\bar{x}(t)|| \leq \tilde{\alpha}(||\bar{x}_0||, t) + \alpha_u (||u||_{\infty})
    \end{equation*}
    Clearly, this implies
    \begin{align*}
        ||x||^2 + ||z||^2 \leq \Big(\tilde{\alpha}(||\bar{x}_0||, t) + \alpha_u (||u||_{\infty})\Big)^2
    \end{align*}
    where each term in the left hand side is smaller than the right hand side. Hence, we can directly conclude that $\exists \tilde{\alpha}_1, \tilde{\alpha}_z \in \mathcal{KL}$ and $\alpha_u \in \mathcal{K}$ s.t. $\forall x_0 \in \mathcal{S}, \forall 0 \leq z_0 < \infty, \ \forall u \in \ell^{\infty}$ and $t\geq0$,
    \begin{align*}
                 ||x(t)|| & \leq \tilde{\alpha}(||x_0||, t) + \tilde{\alpha}(||z_0||,t) + \alpha_u(||u||_{\infty}) \\ 
                 & \leq \tilde{\alpha}(||x_0||, t) + \tilde{\alpha}\big(\alpha_{z_0}(||x_0||),t\big) + \alpha_u(||u||_{\infty}) \\ 
                 & \leq \tilde{\alpha}_1(||x_0||, t) + \alpha_u(||u||_{\infty})
    \end{align*}    
    and
    \begin{align*}
                ||z(t)||  & \leq \tilde{\alpha}(||z_0||, t) + \tilde{\alpha}(||x_0||,t) + \alpha_u(||u||_{\infty}) \\ 
                & \leq \tilde{\alpha}(||z_0||, t) + \tilde{\alpha}(\alpha_{z_0}^{-1}(||z_0||),t\big) + \alpha_u(||u||_{\infty}) \\ 
                 & \leq \tilde{\alpha}_z(||z_0||,t) + \alpha_u(||u||_{\infty}) \leq \alpha_z(||z_0||) + \alpha_u(||u||_{\infty})
    \end{align*}   
    Hence, the system is ISS and ISSf (the BaS is ISS which is a stronger condition than \eqref{eq: input-to-state safety condition}), which completes the proof.
\end{proof}

%The following theorem (similar to that of \cite[Theorem~4.19]{khalil2002nonlinear}) provides a sufficient condition for input-to-state safety stability (IS$^3$).
%\begin{theorem} \label{theorem: sufficient condition for IS3}
%    Consider a continuously differentiable function $V(\bar{x}): \bar{\mathcal{X}} \rightarrow \mathbb{R}$,  a continuous positive definite function $W(\bar{x}):\bar{\mathcal{X}} \rightarrow \mathbb{R}$, class $\mathcal{K}_{\infty}$ functions $\alpha_1, \alpha_2$ and a class $\mathcal{K}$ function $\rho$ such that
%    \begin{equation}
%        \alpha_1(||\bar{x}||) \leq V(t,\bar{x}) \leq \alpha_2(||\bar{x}||) 
%    \end{equation}
%    \begin{equation}
%        V_{\bar{x}} f(\bar{x},u) \leq -W(\bar{x}), \ \forall ||\bar{x}|| \geq \rho(||u||) > 0
%    \end{equation}
%    for all $(\bar{x},u) \in \bar{\mathcal{X}} \times \mathbb{R}^m$. Then, the safety-critical system \eqref{eq: control system dynamics} is input-to-state safely stable (IS$^3$) with $\alpha_u = \alpha_1^{-1} \circ \alpha_2 \circ \rho$. 
%\end{theorem}

%\begin{proof}
%By \cite[Theorem~4.19]{khalil2002nonlinear}, we have 
%\begin{equation}
%    ||\bar{x}(t)|| \leq \tilde{\alpha}(||\bar{x}_0||, t) + \alpha_u (||u||_{\infty})
%\end{equation}
%for all $\bar{x}_0 \in \mathcal{S}\times\mathcal{B}, \ u \in \ell^{\infty}$ and $t\geq 0$. Then, by \autoref{prop: safety system is ISS iff orignal is IS3}, the system \eqref{eq: control system dynamics} is input-to-state safely stable (IS$^3$). 
%\end{proof}

Since our goal is to design multi-objective safe controllers, we are particularly interested in ensuring that the safety embedded controller renders the safe set robustly invariant under bounded input disturbances.\autoref{prop: safety system is ISS iff orignal is IS3} is pivotal in this context, as it enables the use of established ISS tools to characterize and enforce IS$^3$ through the safety embedded system. The Lyapunov-like theorem in \cite[Theorem~4.19]{khalil2002nonlinear}, provides a sufficient condition for ISS. Applying that theorem for the safety embedded system, allows us to construct a sufficient condition for IS$^3$ based on \autoref{prop: safety system is ISS iff orignal is IS3}.

Moreover, under standard regularity conditions, global exponential stability of the undisturbed system implies ISS \cite{khalil2002nonlinear}, while local asymptotic stability implies local ISS \cite{mironchenko2016local-iss}. In other words, an asymptotically stable autonomous system is inherently robust to small bounded disturbances \cite[Chapter~2]{mironchenko2023ISS-book}, as ISS generalizes uniform asymptotic stability by accounting for external inputs. The following corollary leverages these facts to establish sufficient conditions for IS$^3$ and ISSf under the proposed framework, which is a direct result of \autoref{prop: safety system is ISS iff orignal is IS3} and \cite[Theorem~4.19]{khalil2002nonlinear}.

%Under general regularity conditions of the systems dynamics in \eqref{eq: control system dynamics}, it has been established that global exponential stability of the autonomous, i.e. undisturbed, system corresponds to input-to-state stability \cite{khalil2002nonlinear}. Consequently, local asymptotic stability of the autonomous system corresponds to local input-to-state stability \cite{mironchenko2016local-iss}. In other words, an asymptotically stable undisturbed system is intrinsically robust to small enough external inputs \cite[Chapter~2]{mironchenko2023ISS-book}. This stems from the definition of input-to-state stability and the fact that asymptotic stability is equivalent to uniform asymptotic stability of autonomous systems. The following corollary provides neat results for ISSf and IS$^{3}$ when combined with the developed results in the preceding subsections, which is a direct result of \autoref{prop: safety system is ISS iff orignal is IS3} and \cite[Theorem~4.19]{khalil2002nonlinear}.

\begin{corollary} \label{corollary: linear system AS then IS3}
    Suppose that the safety embedded system \eqref{eq: safety embedded system} with $\bar{f}(0,0)=0$ is continuously differentiable. Let $u=K(\bar{x})$ be a continuous feedback controller such that the undisturbed closed-loop system $\dot{\bar{x}} = \bar{f}\big(\bar{x},K(\bar{x})\big)$ is globally Lipschitz for all $\bar{x} \in \mathcal{S} \times \mathcal{B}$ and its origin $\bar{x}=0$ is globally exponentially stable. Then, the safety-critical system \eqref{eq: control system dynamics} is input-to-state safely stable (IS$^3$). It is locally IS$^3$, if the continuous feedback controller renders the origin of the undisturbed closed-loop system uniformly asymptotically stable.
\end{corollary}
\begin{proof}
    The proof follows that of \cite[Lemma~4.6]{khalil2002nonlinear} for ISS in the global sense (see \textcite[Corollary~2.39]{mironchenko2023ISS-book}, for the local case), and  by \autoref{prop: safety system is ISS iff orignal is IS3}, the safety-critical system \eqref{eq: control system dynamics} is (locally) IS$^3$.  
\end{proof}

Such results are particularly useful for safe stabilization, especially when employing locally asymptotically stabilizing feedback control methods. This will be demonstrated in the following section through an application example using linear feedback control.

\vspace{-3mm}
\subsection{Robust Safe Feedback Control}
Given a feedback controller with some external disturbance $d \in \mathcal{D} \subset \ell^{\infty}$, i.e. $u=K(\bar{x}) + d(t)$, the closed-loop system is then of the form
\begin{equation} \label{eq: control system with disturbance d}
    \dot{x} = f\big(x, K(\bar{x}) + d \big)
\end{equation} 
We are interested in investigating \textit{feedback} controllers that are functions of the BaS in ensuring or improving ISSf, at least locally, through providing the controller with the \textit{safety status} of the system. What we mean by that is that the feedback controller would react aggressively when the BaS value is high which means that the system is close to the unsafe region and would react meekly when the BaS value is low which means that the system is away from the unsafe region. By applying the proposed ISSf and IS$^3$ conditions, one can design a robust, safe feedback control law that either accommodates known external disturbance bounds or manages disturbances up to a specified level. % A Lyapunov based feedback control \textit{redesign} is possible as proposed in \cite{romdlony2016newnotionISSf} based on the work of \textcite{sontag1989smooth} (also see \cite[Chapter~5, Theorem~5.4]{mironchenko2023ISS-book}). 

%\textcolor{red}{Introduce Lyap based controller for IS$^3$ as \cite{romdlony2016newnotionISSf} also see \cite[Chapter~5, Theorem~5.4]{mironchenko2023ISS-book}??????}

Let us consider the following instructive example as a case study, which is used in \cite{kolathaya2018ISSfCBF} to motivate the construction of input-to-state safe barrier functions. 

\subsubsection*{Case Study}
The dynamical system 
\begin{equation*}
    \dot{x}=-x+x^2u
\end{equation*}
is subject to the safety condition defined by $\mathcal{S}=\{x\in \mathbb{R}| 2-x >0\}$. We would like to verify if the system is safe despite some actuator disturbance and derive a condition based on \eqref{eq: input-to-state safety condition}. %For the moment, let $u=d(t)$, i.e. a pure input disturbance. 
Choosing an inverse barrier function, %, whose dynamics can be derived to be 
%\begin{equation} \label{eq: BaS dynamics for ISSf scalar example}
%    \dot{z} = \frac{-1}{h^2}h_x \dot{x} = (z+0.5)^2 (-x+x^2 u)
%\end{equation}
%In such a case, 
the safety embedded system is
\begin{equation} \label{eq: ISSf case study safety embedded system}
\begin{split}
%    & \dot{x}=-x+x^2K(x,z)+x^2d \\
%    & \dot{z} = (z+0.5)^2 (-x+x^2K(x,z)+x^2 d)
& \dot{x}=-x+x^2u \\
    & \dot{z} = (z+0.5)^2 (-x+x^2u)
\end{split}
\end{equation}
Using an initial condition of BaS to be $z(0)=\mathbf{B}\circ h\big(x(0)\big) - \mathbf{B}\circ h(0) = \frac{1}{2-x(0)}-\frac{1}{2}$, we can derive the BaS to be $z=\frac{0.5 x}{2-x}$. %Note that we dropped the $\gamma$ term since it is only a regularity term that is not needed for the analysis in this case. 
It can be easily verified that this system is globally asymptotically stable and safe when $u=0$. Nonetheless, It can be seen that for any initial condition such that $|xu| \geq 1$, the system's trajectories grow unbounded, i.e. neither stability nor safety is guaranteed. %$\forall x_0 \in \mathcal{S}$ the BaS in \eqref{eq: BaS dynamics for ISSf scalar example} is uniformly (asymptotically) stable $\forall |z| \leq \frac{1}{2(2|u|-1)}, \ \forall t \geq 0$, or equivalently $\forall |x| \leq \frac{1}{|u|}, \ \forall t \geq 0$. Therefore, the system is ISSf with $\alpha(r)=\frac{1}{2(2|r|-1)}$. In fact, the system is ISS $\forall |x| \leq \frac{1}{|u|}$, i.e. with $\alpha(r)=\frac{1}{|r|}$ and hence the system is IS$^3$ $\forall x_0 \in \mathcal{S}$ and $\forall t\geq0$.
Pursuing this further, we are interested in designing a feedback controller that insures, or improves, the ISSf property of the system given the existence of such a feedback control (as discussed in \autoref{remark: ISSf feedback control design}). Using the BaS embedding technique, the feedback controller can be a function of the BaS $z$. This is of a major benefit as the controller reacts proactively to push the system away from the unsafe region and will react more aggressively if the systems starts or gets closer to the unsafe region. It is worth mentioning that one option here is to select the controller to be the negative of the gradient of the safety function, i.e. zeroing barrier function, $K(x)=-x^2$ as suggested by \textcite{kolathaya2018ISSfCBF} which provides a \textit{local} ISSf notion. Nonetheless, it is worth noting that such a controller may be overly aggressive and destabilizes the system when the system is a way from the unsafe region, e.g. when $x_0\ll 0$. Furthermore, although it has the advantage of moving the system from the unsafe region to the safe region \cite{kolathaya2018ISSfCBF}, it still allows for constraint violation under some large disturbances. Our proposition is to consider a linear feedback of the BaS, i.e. $u=K(x,z) = -K_z z$. Note that such a linear feedback law, if it indeed satisfies the required conditions, corresponds to some Lyapunov function (by Lyapunov converse theorems \cite{khalil2002nonlinear}). Additionally, since we are after the ISSf property, we will deal with the BaS only at first. Under the disturbed feedback controller, $u=K(x,z)+d(t)$, we have
\begin{equation} \label{eq: ISSf case study safety embedded system with controller}
\begin{split}
%    & \dot{x}=-x+x^2K(x,z)+x^2d \\
    & \dot{z} = (z+0.5)^2 (-x+x^2K(x,z)+x^2 d)
\end{split}
\end{equation}
%Specifically, the BaS equation is given by
%\begin{equation} \label{eq: ISSf case study BaS with safe feedback}
%    \dot{z} = (z+0.5)^2 (-x+K_z x^2 z+x^2 d)
%\end{equation}

\begin{figure}[t]
        \centering
    \includegraphics[trim=20 0 20 0, clip, width=0.9\linewidth]{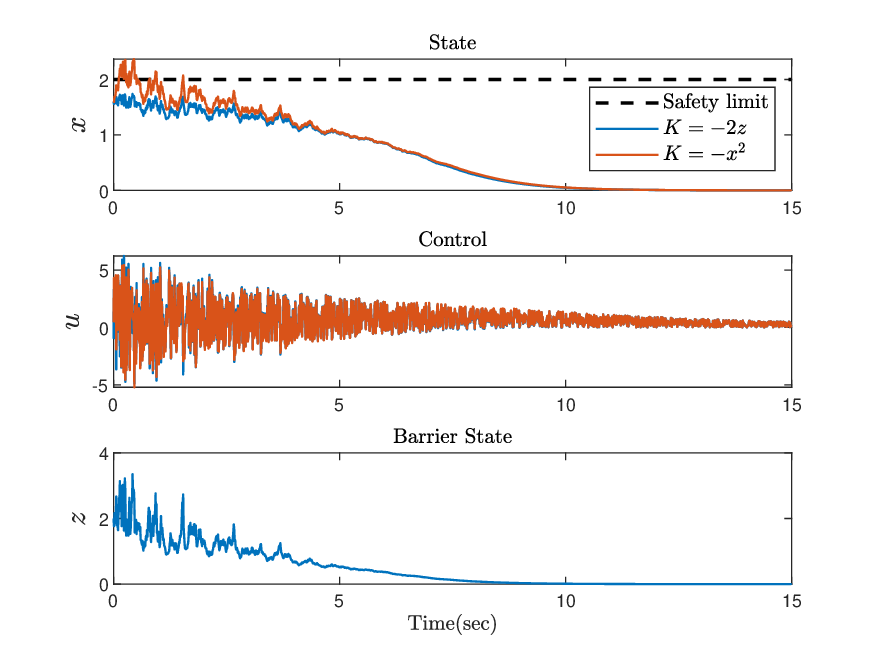}
    \vspace{-3mm}
      \caption{Simulations of the closed-loop system under the BaS feedback controller $K=-2z$ and a state feedback controller $K=-x^2$. Shown here are the state's progression, the control action and the BaS progression over time.}
      \label{fig: ISSf case study example}
      \vspace{-6mm}
\end{figure}

Let us consider $V=\frac{1}{2}z^2$. Then, the derivative of $V$ along the BaS trajectories is
\begin{equation*}
    \dot{V} = z (z+0.5)^2 (-x+x^2K(x,z)+x^2 d)
\end{equation*}
which can be rewritten as
\begin{equation*}
    \dot{V} = 4 K(x,z) z^3 + (4d-2) z^3  - z^2
\end{equation*}
Assuming a negative feedback $K(x,z)=-K_z z$, $K_z>1$, of the barrier state as the controller yields
\begin{equation*}
     \dot{V} = - 4K_z z^4 + (4d-2) z^3  - z^2 \leq -8V(z)p(z)
\end{equation*}
%It is worth mentioning that one option here is to select the controller to be the negative of the gradient of the barrier with respect to $u$, i.e. $K(x,z) = -z^2$, which is the suggested control law by \textcite{romdlony2019robustness} and is used for the zeroing barrier function for this example by \textcite{kolathaya2018ISSfCBF}, i.e. $K(x)=-x^2$. Nonetheless, it is worth noting that such a controller may be overly aggressive. In the case of $K(x)=-x^2$, it could be too aggressive and destabilizes the system when the system is a way from the unsafe region, e.g. when $x_0\ll 0$. Furthermore, although it has the advantage of moving the system from the unsafe region to the safe region, it still allows for constraint violation under some large disturbances. On the other hand, the BaS feedback controller $K(x,z) = -z^2$, could be too aggressive when it is very close to the unsafe region which may not be desirable.
%This implies that we can ensure uniform stability of the BaS if we design $K_z$ such that
%\begin{align*}
%    K_z x^2 z + x^2 ||d||_{\infty} - x \leq 0 \Rightarrow K_z \leq \frac{1-||d||_{\infty}}{zx}
%\end{align*}
where $p(|z|):=K_z|z|^2-(|d|+0.5)|z|+0.25$. It can be seen that if $|d|\leq 0.5$, the quadratic polynomial $p(|z|)$ is minimized at $|z|=z_{min}=\frac{|d|+0.5}{2K_z}$ and 
$$
p(z_{\min})=0.25(K_z-(|d|+0.5)^2)/K_z\geq 0.25(K_z-1)/K_z
$$
Selecting $|z|\geq \frac{|d|+0.5}{K_z}$ on the other hand guarantees that $p(|z|) \geq  0.25 > 0.25(K_z-1)/K_z$) for $|d|\geq 0.5$. The function
$$
\alpha_u(|d|) := \begin{cases} 2 |d|/K_z & , \ |d| <0.5 \\
(|d| + 0.5)/K_z & , \ |d| \geq 0.5
\end{cases}
$$
is continuous, strictly increasing, and zero at $d=0$, hence a class ${\mathcal{K}}$ function. Furthermore, as was shown above, for $|z|> \alpha_u(|d|)$, we have $p(z)>0.25(K_z-1)/K_z>0$, and consequently $\dot{V} (z) \leq -2(K_z-1)V(z)/K_z$, $\forall |d|\geq 0$
Therefor, by \cite[Theorem~4.19]{khalil2002nonlinear}, 
\begin{equation*}
     ||z(t)|| \leq \tilde{\alpha}_z(||z_0||, t) + \alpha_u (||d||_{\infty})
\end{equation*}
and by \autoref{def: ISSf}, the system is ISSf with $\alpha_u(d)$ as defined above.
%In fact, the origin of the barrier state is exponentially stable. 
Hence, when $d(t) \rightarrow 0$ as $t\rightarrow \infty$, $z(t) \rightarrow 0$. Consequently, for this system, the state goes to the origin. In fact, the safety embedded system \eqref{eq: ISSf case study safety embedded system} can be shown to be ISS using the bounds derived above for the barrier state and the disturbance. Therefore, the system is input-to-state safely stable (IS$^3$) by \autoref{prop: safety system is ISS iff orignal is IS3}. It  must be noted, however, that the derived bounds are conservative as the used Lyapunov-like theorem \cite[Theorem~4.19]{khalil2002nonlinear} provides a sufficient, but not necessary, condition. 

\autoref{fig: ISSf case study example} shows a random run for the problem with an initial condition $x_0=1.6$ and a relatively large, uniformly distributed and decreasing random disturbance, with $||d||_{\infty}=9.585$. Note that even if the disturbance is not convergent to zero, the proposed feedback controller still ensures ISSf of the system since it is enough to have the BaS being uniformly stable and not necessarily uniformly asymptotically stable. It is worth mentioning that we designed the controller and then quantified the ISSf bounds but knowing the bounds of the disturbance can also help us designing a safe feedback controller with a desired robustness performance.

\begin{remark}
    One can derive parallel arguments to the proposed development through Lyapunov analysis and Lyapunov ISSf arguments which also allow us then to design barrier-Lyapunov based feedback controllers using a gradient-based control such as Sontag's Lyapunov \textit{redesign} or Sontag's universal control law \cite{sontag1989smooth,Sontag2008issconcepts,sontag1996new} (see also \cite[Chapter~5, Theorem~5.4]{mironchenko2023ISS-book}) as done in \cite{romdlony2016newnotionISSf,kolathaya2018ISSfCBF}. Indeed, one can directly use Lyapunov converse and converse-like theorems to achieve the same established results above. It is definitely interesting to investigate the developments of Lyapunov arguments and control Lyapunov functions in the context of barrier states and safety embedded systems and the connection to the work in \cite{romdlony2016stabilization,romdlony2019robustness,ames2016CBF-forSaferyCritControl,tee2009barrier_Lyap_automatica}. Nonetheless, such investigations and extensions require thorough examination and are better suited for separate study. %Therefore we sought to avoid the use of Lyapunov arguments and Lyapunov based controls in our development in this paper which we leave to be studied in a separate work.
\end{remark}

\vspace{-3mm}
\section{Applications Examples} \label{sec: Applications Examples}

\subsection{Safety Embedded Linear Stabilization} \label{subsec: Safety Embedded Linear Stabilization}
In this example, let us consider the constrained and open-loop unstable system
\begin{equation*}
    \dot{x} = \begin{bmatrix} \dot{x}_1 \\ \dot{x}_2 \end{bmatrix}= \begin{bmatrix} 1 & -5 \\ 0 & -1\end{bmatrix} \begin{bmatrix} x_1 \\ x_2 \end{bmatrix} + \begin{bmatrix} 0 \\ 1\end{bmatrix} u
\end{equation*}
subject to the coupled nonlinear safety condition
\begin{equation*}
    h(x) = (x_1 - 2)^2 + (x_2 - 2)^2 - 0.5^2 >0
\end{equation*}
and the performance objective of closed-loop system's poles at $-2$ and $-3$. 

To solve this multi-objective problem through the proposed technique, we first select the inverse barrier, and therefore $\beta = \frac{1}{h(x)}$. Using the proposed technique, with $\gamma=1$ the barrier state equation is given by
\begin{equation}
    \dot{z} = -(z+\beta_0)^2 [2(x_1-2) \ \ 2(x_2-2)] \begin{bmatrix} x_1 -5x_2  \\ -x_2 + u  \end{bmatrix} - \big(z+\beta_0 - \beta(x) \big)
\end{equation}
% or
%\begin{equation}
%    \dot{z} = -(z+\beta_0)^2 \big(x_1^2 - x_2^2 - 5 x_1 x_2 -2x_1 + 12 x_2 +x_2 u -2u\big) - \big(z+\beta_0 - \beta(x) \big)
%\end{equation}
Augmenting this BaS to the system and linearizing the safety embedded system, yields the safety embedded linearized system
\begin{equation}
    \dot{\bar{x}}=  \begin{bmatrix} 1 & -5 & 0\\ 0 & -1 & 0 \\ \frac{8}{7.75^2} & \frac{-20}{7.75^2} & -1 \end{bmatrix}  \bar{x} +   \begin{bmatrix} 0 \\ 1 \\ \frac{4}{7.75^2} \end{bmatrix} u
\end{equation}
where $\bar{x}=[x_1 \ \ x_2 \ \ z]^{\top}$.
To achieve the performance objectives such that the system's states have the closed-loop poles at $-2$ and $-3$, we resort to using the well known pole placement method. It is important to note that the pole-placement solution is not unique, i.e. there might be other solutions that can place the poles at the desired locations but provides an unsafe behavior or a weak performance, e.g. sensitive solution to nonlinearities and hence the local solution provides a small region of attraction. We use the robust pole assignment method in \cite{kautsky1985robustpole} used by \verb|MATLAB| command \verb|place|, which provides a solution with a minimized sensitivity with a maximized stability margin. This algorithm is very helpful in our context of safety and linear control for an inherently nonlinear control problem. The control law to \textit{safely} stabilize the system is computed to be 
\begin{equation} \label{eq: safety embedded pole placement control}
u = 2.1143 x_1 - 5.2857 x_2 + 4.2902 z
\vspace{-2mm}
\end{equation}

It is worth noting that although the controller is linear with respect to the linearized safety embedded system, it is a nonlinear function of the original state $x$. A phase portrait of the closed-loop system in the $x_1 \times x_2$ state space is shown in \autoref{fig: linear control example - safe stabilization}. It can be seen that the trajectories and the closed-loop vector field turn around to avoid the unsafe region shown as a red circle. This shows that indeed the designed linear controller is able to \textit{safely stabilize} the system successfully achieving safety and performance objectives.

\begin{figure}[t]
    \centering
    \includegraphics[trim=20 0 20 0, clip, width=0.9\linewidth]{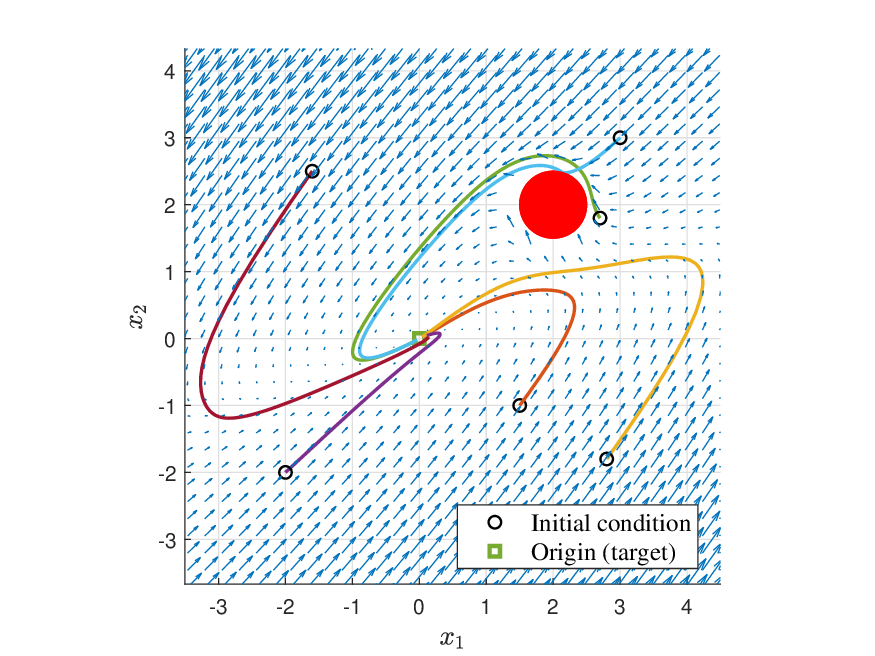}
    \\
    \includegraphics[trim=20 0 20 0, clip, width=.7\linewidth]{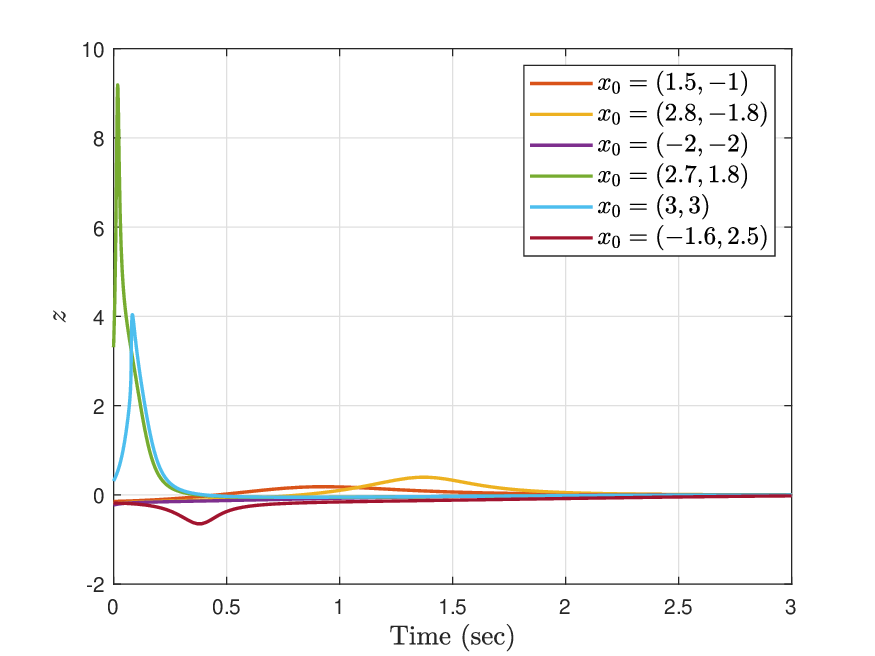}
    \caption{A phase portrait of the closed-loop system for the stabilization example under the safety embedded controller $u = 2.1143 x_1 - 5.2857 x_2 +4.2902 z$ along with some trajectories starting from different initial conditions. The bottom figure shows BaS progression over time for the different trajectories.}
      \label{fig: linear control example - safe stabilization}
      \vspace{-6mm}
\end{figure}

\vspace{-2mm}
\subsection{Input Constrained Safety Embedded Linear Stabilization} \label{subsec: Input Constrained Safety Embedded Linear Stabilization}
In practice, control power is limited by the physics of the system or by design for some objectives, i.e. $\mathcal{U}$ is bounded. Let us consider the preceding example with some input constraints. Specifically, the admissible control set is bounded by $(-5,5)$. %This can be translated as 
%\begin{equation*}
%    -5 < u < 5
%\end{equation*}
Then, the safety condition can be written as 
\begin{equation*}
    g_1(u) = 5 - u > 0 \qquad , \ g_2(u) = u+5 > 0
\end{equation*}
Note that one could define a single quadratic function for the constraint and then define a single BaS but we choose to define two linear constraints to define multiple BaS, i.e. a BaS for each constraint, for illustration purposes. Following the proposed technique in \autoref{subsec: input BaS} and the same procedure in the foregoing example, we define three barrier states $z_1, \ z_2$ and $z_3$ for $g_1, \ g_2$, and $h$ respectively with $\gamma_1=1.8, \ \gamma_2 = 1.4, \ \gamma_3=1$. The linearized safety embedded system is given by 
\begin{equation*}
    \dot{\bar{x}}=  \begin{bmatrix} 1 & -5 & 0 & 0 & 0 & 0 \\ 0 & -1 & 1 & 0 & 0 & 0 \\ 0 & 0 &0 & 0 & 0 & 0 \\ 0 & 0 & \frac{1.8}{5^2} & -1.8 & 0 & 0 \\ 0 & 0 & -\frac{1.4}{5^2} & 0 & -1.4 &  0 \\ 
    \frac{8}{7.75^2} & \frac{-20}{7.75^2} & \frac{4}{7.75^2} & 0 & 0 & -1 \end{bmatrix}  \bar{x} +
    \begin{bmatrix} 0 \\ 0 \\ 1 \\ \frac{1}{5^2} \\ -\frac{1}{5^2} \\ 0 \end{bmatrix} v
\end{equation*}
where $\bar{x} = [x_1, \ x_2, \ u, \ z_1, \ z_2, \ z_3]^{\top}$. Pole placement is used again to achieve the performance objective of having closed-loop poles of $-2$ and $-3$. Note that now we have the new control state equation $\dot{u} = v$ which we need to assign a pole for. We select the desired pole to be $-10^3$ for numerical stability. It is noteworthy that in this case the relative degree of the system with respect to the barrier function corresponding to the obstacle is elevated. That does not affect our formulation, however, and the system is still stabilizable. The computed control law required to stabilize the aforementioned embedded system with the desired closed-loop poles is $u = -K \bar{x}$ where
\begin{equation} \label{input limits aware control example}
    K = [-2.12   \  5.29 \  -.101 \   2.67  \ -24.97 \  -4.29] \times 10^3
\end{equation}
\autoref{fig: linear control example - input constraint} shows the closed-loop response starting from some initial conditions at which excessive control inputs beyond the allowed limit is needed to safely stabilize the system in the preceding example without input constraints. The designed controller effectively stabilizes the system while adhering to control constraints. It executes conservative, yet secure, actions which result in a slightly longer time to achieve stabilization. 

\begin{figure}[t]
    \centering
    \includegraphics[trim=20 0 20 0, clip, width=.7\linewidth]{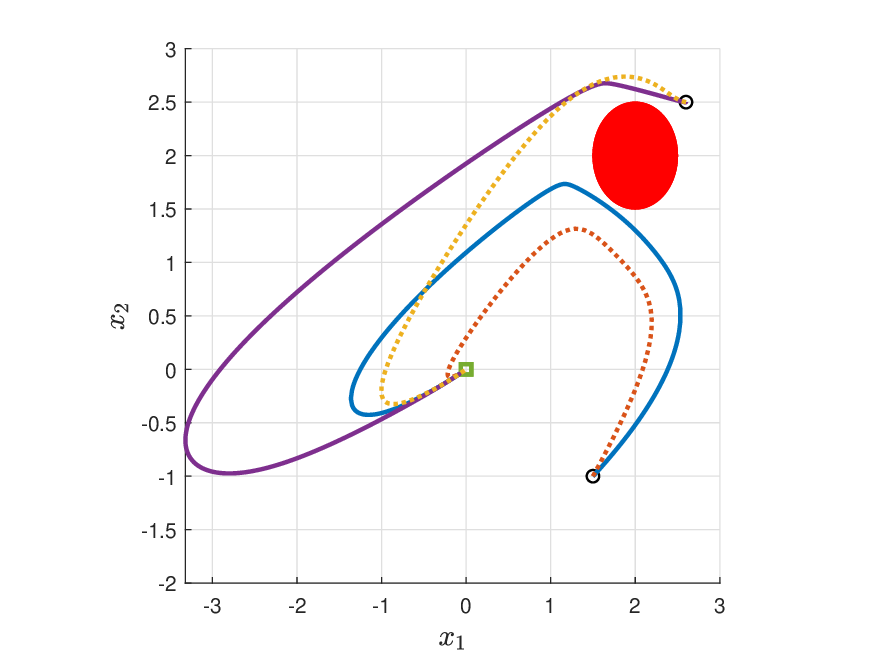}
    \\
    \includegraphics[trim=20 0 20 0, clip, width=.8\linewidth]{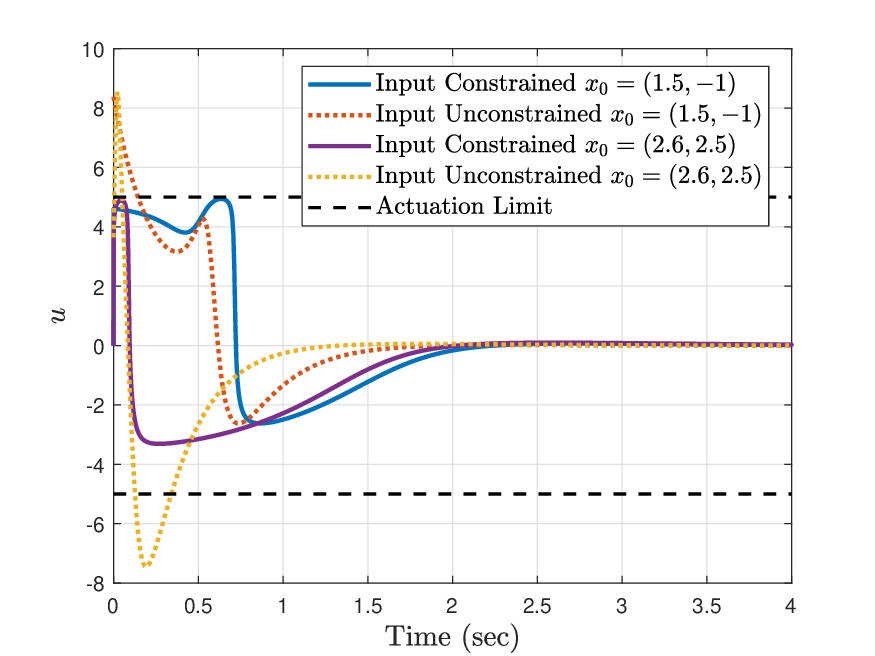}
    \caption{Numerical simulations of the closed-loop system under the input-limits aware safety embedded controller in \eqref{input limits aware control example} (solid) compared against the safety embedded controller \eqref{eq: safety embedded pole placement control} (dotted) that does not take input constraints into account.}
      \label{fig: linear control example - input constraint}
      \vspace{-6mm}
\end{figure}

\vspace{-3mm}
\subsection{Adaptive Cruise Control via Barrier States}
With a cruise control system, the equipped automobile is to maintain a desired steady speed automatically regardless of other factors such as distance to vehicles ahead. With an adaptive cruise control (ACC), an advanced driver-assistance system (ADAS), the control system maintains the desired speed when possible and automatically adjust the speed when needed to maintain a safe distance from vehicles ahead. ACC systems are very popular and a key part of essential low level automating systems in autonomous vehicles.

The safety constraint for this problem can be described through the time-headway, i.e. the time interval between successive vehicles passing a point in a traffic stream. For this application, we adopt the problem in \cite{ames2016CBF-forSaferyCritControl}, i.e. the dynamics, the performance objective, the safety condition and the parameters, etc. The safety constraint is formulated as 
\begin{equation}
    \frac{D}{v_f} > \tau_d
\end{equation}
where $D$ is the distance between two consecutive vehicles (meters), $v_f$ is the velocity of the following, i.e. controlled, vehicle (m$/$s), and $\tau_d$ is the desired time headway that is set to be $1.8$ seconds according to the rule of thumb of allowing a minimum distance of half the speedometer of the following vehicle \cite{vogel2003comparison_headwaytime}. 

For cruise control systems, various robust PID control schemes were deployed in the early 1980s when the microprocessor technology revolutionized vehicles' systems \cite{shaout1997cruise}. PID controls are commonly used in various other vehicle control systems due to their simplicity and ease of achieving performance characteristic such as response time, rise time, steady state error, etc. and their fast and stable responses. For the ACC problem,  we develop a PID controller for the cruise control problem with barrier states augmentation to develop an adaptive cruise controller. We show that the designed controller, termed proportional-integral-derivative-barrier (PIDB), is capable of simultaneously satisfying safety and stability (performance) conditions when possible and can enforce the safety constraint, by design, when the two objectives are not simultaneously achievable. It is worth noting that one could design a PIB controller, i.e. without the derivative control element. We choose to include the derivative controller for mainly two reasons. The first is to provide a smooth transient especially when the vehicle needs to dramatically changes its speed for safety. The second is to show that the proposed method works well even with systems with a high-relative degree without a need of further considerations. 

The system considered in this problem consists of a following vehicle and a leading one in which the following vehicle is the controlled one with an ACC system to cruise with a desired speed while ensuring a safety distance from the leading vehicle. The dynamics of the system is given by
\begin{equation}
    \begin{bmatrix} \dot{v}_l \\  \dot{v}_f \\ \dot{D} \end{bmatrix} = \begin{bmatrix}
        a_l \\ \frac{1}{M_f} (F -f_0 - f_1 v_f - f_2 v^2_f) \\ v_l - v_f
    \end{bmatrix}
\end{equation}
where $v_l$ is the leading vehicle's velocity (m$/$s), $a_l$ is the leading vehicle's acceleration (m$/$s$^2$), $M_f$ is the following vehicle's mass (kilograms), $f_0$, $f_1$ and $f_2$ are the aerodynamic drag constants, and $F$ is the following vehicle's wheel force (Newtons), i.e. the control input. 

\begin{figure}[t]
    \centering
    \hspace{-.45cm}
    \includegraphics[trim=15 0 32 20, clip, width=.51\linewidth]{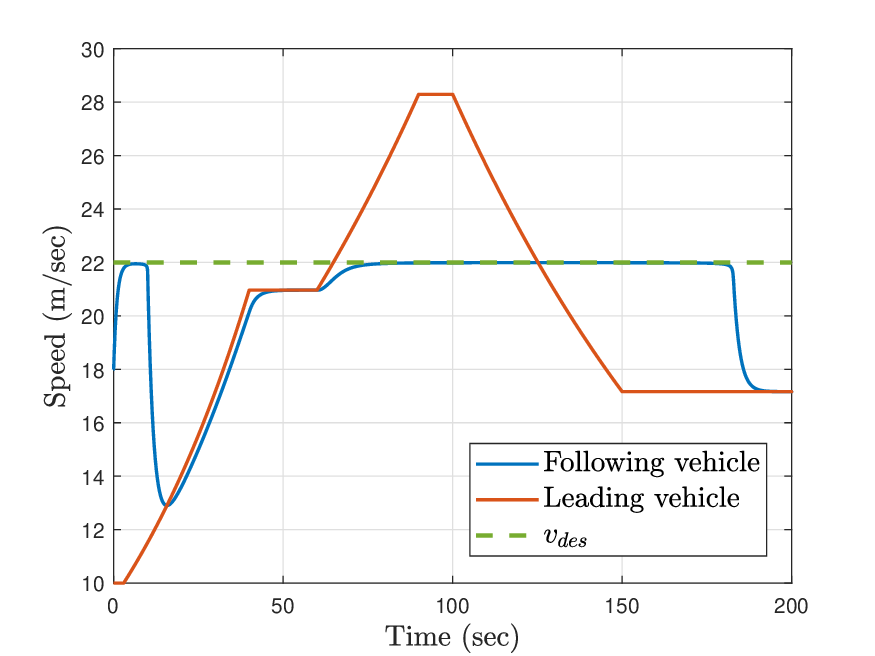}
    \includegraphics[trim=17 0 32 20, clip, width=.51\linewidth]{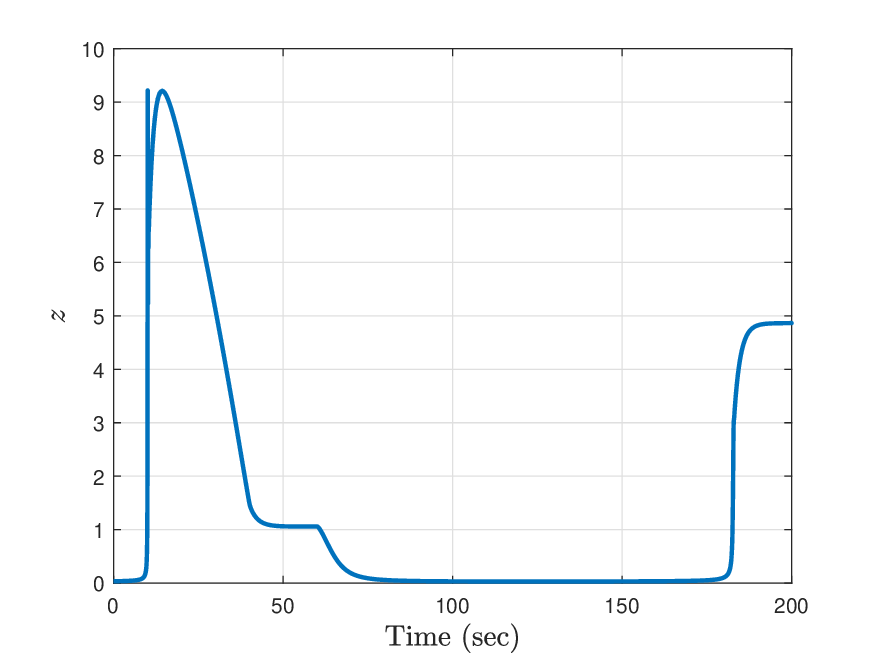}
    \\
    \hspace{-.45cm}
    \includegraphics[trim=15 0 32 20, clip, width=.51\linewidth]{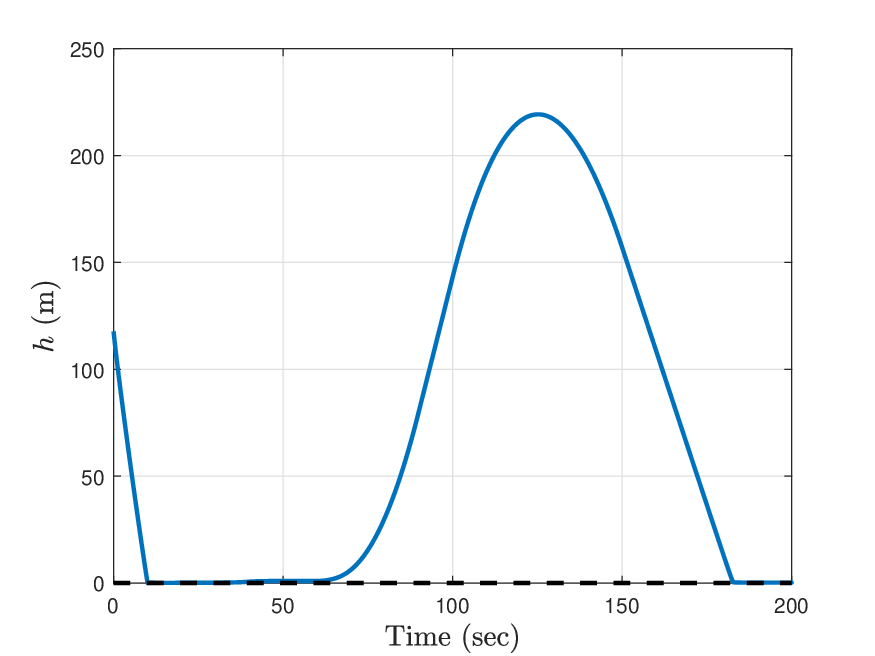}
    \includegraphics[trim=15 0 32 20, clip, width=.51\linewidth]{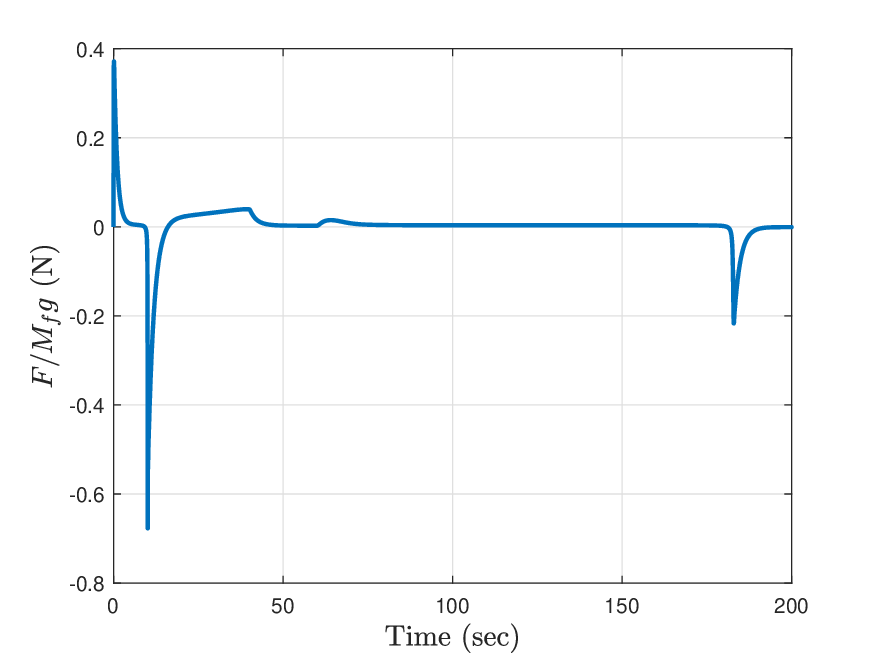}
    \caption{Adaptive cruise control (ACC) simulation results derived by a PIDB controller with high proportional and derivative control gains. Top-left figure shows the speed of the following, controlled, vehicle (blue) driving at the desired speed (dashed green) when possible and slowing down to ensure a safe time-headway (bottom-left) from the leading vehicle. Top-right figure shows the profile of the barrier state over time which stays bounded indicating forward invariance of the safe set. Bottom-right figure shows the controlled vehicle acceleration as a function of the Earth gravitational force $g$.}
      \label{fig: ACC PIDB - large gains}
      \vspace{-6mm}
\end{figure}

The performance objective is to regulate the leading vehicle's speed at a desired velocity $v_d$. 
To achieve the performance objective and ensure robustness to parametric perturbations, we introduce an integral action by augmenting the dynamics with an integral state of the regulation error, $e = v_f - v_d$. Consequently, derivative action is introduced by augmenting the rate of change of the error, $ \dot{e} $, whose dynamics correspond to the acceleration of the following vehicle, 
\begin{equation} 
\ddot{v}_f = \frac{1}{M_f} ( u - f_1 \dot{v}_f - f_2 v_f \dot{v}_f )
\end{equation}
where $u = \dot{F}$ which we need to integrate then to produce the control input $F$. The problem now can be expressed by the following dynamics,
\begin{equation}
   \dot{x} = f(x,u) = \begin{bmatrix}
        a_l \\ x_5 \\ x_1 - x_2 \\ x_2 - v_d \\ \frac{1}{M_f} (u - f_1 x_5 - f_2 x_2 x_5)
    \end{bmatrix}
\end{equation}
where $x = [x_1, \ x_2, \ x_3, \ x_4, \ x_5]^{\top} = [v_l, \ v_f, \ D, \ \int e, \ \dot{v}_f]^{\top}$, subject to the safety condition 
\begin{equation}
    h = D - \tau_d v_f = x_3 - \tau_d x_2 > 0 
\end{equation}

To enforce the safety constraint, we introduce the barrier $\beta = \frac{1}{h}$, i.e. $\mathbf{B}$ is the inverse barrier function. Then, we define the barrier state $z$ as developed in \autoref{sec: Barrier States} and derive the state equation 
\begin{equation*} \begin{split}
    \dot{z} 
    = -(z+\beta_0)^2 (-\tau_d x_5 + x_1 - x_2 ) - (z+\beta_0 - 
    \frac{1}{x_3 - \tau_d x_2} ) 
\end{split} \end{equation*}
It is worth noting that the control input does not show up in the BaS equation, i.e. the relative degree is higher than 1. Nonetheless, no special consideration needs to be taken. Augmenting the BaS to the dynamics results in the safety embedded system 
\begin{equation*}
\hspace{-3mm}   \dot{\bar{x}} = \begin{bmatrix}
        a_l \\ x_5 \\ x_1 - x_2 \\ x_2 - v_d \\ \frac{1}{M_f} (u - f_1 x_5 - f_2 x_2 x_5) \\ -(x_6+\beta_0)^2 (-\tau_d x_5 + x_1 - x_2 ) - (x_6+\beta_0 - 
    \frac{1}{x_3 - \tau_d x_2} )
    \end{bmatrix}
\end{equation*}
where $\bar{x} = [x, \ x_6]^{\top} = [x, \ z]^{\top}$ and $\beta_0 = \frac{1}{h(x_d)}$. For this example, the system's parameters, adopted from \cite{ames2016CBF-forSaferyCritControl}, are $M_f=1650$ kg, $g=9.81$ m$/$s$^2$,  $f_0=0.1$ N, $f_1=5$ N s$/$m, $f_2=0.25$ N s$^2/$m$^2$, the desired speed is set to $22$ m$/$s and the initial conditions are $(v_l(0)=10, v_f(0)=18,D(0)=150)$. 

\begin{figure}[t]
    \centering
    \hspace{-.45cm}
    \includegraphics[trim=15 0 32 20, clip, width=.51\linewidth]{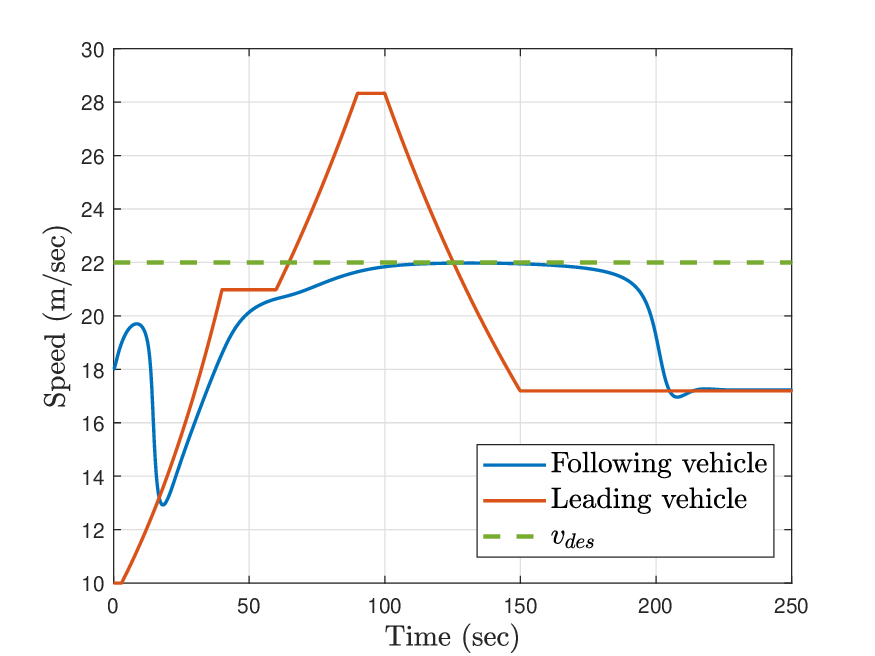}
    \includegraphics[trim=17 0 32 20, clip, width=.51\linewidth]{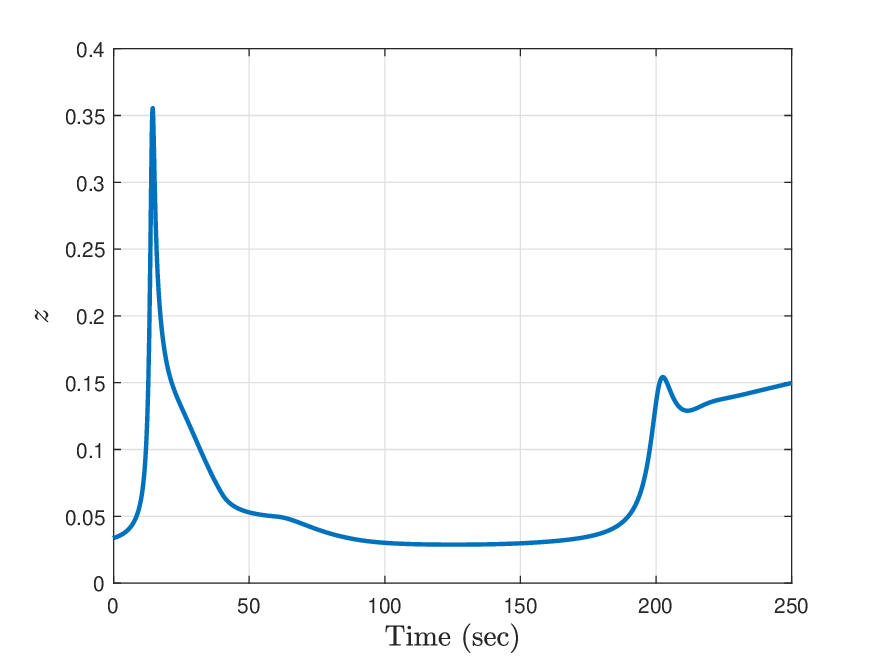}
    \\
    \hspace{-.45cm}
    \includegraphics[trim=15 0 32 20, clip, width=.51\linewidth]{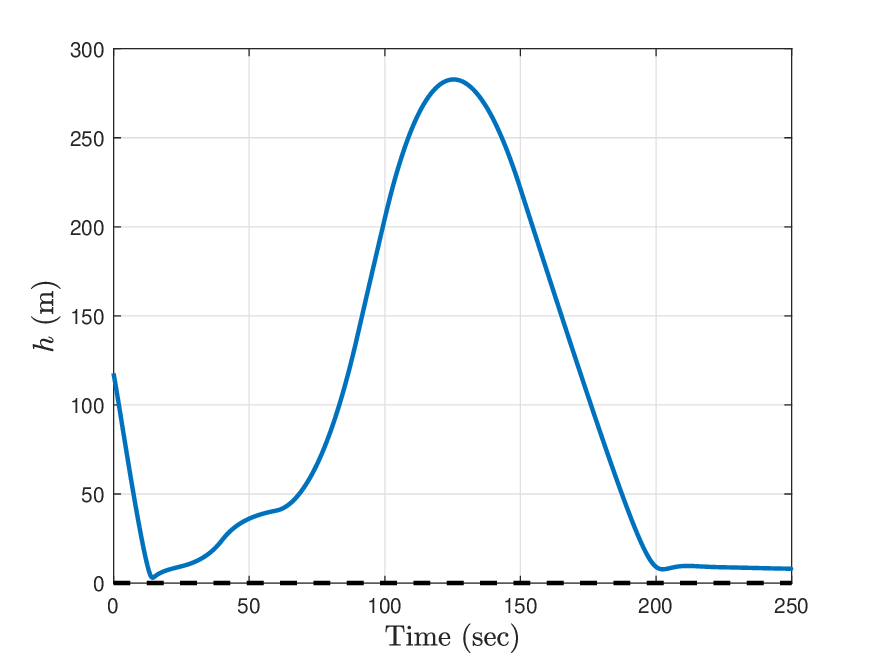}
    \includegraphics[trim=15 0 32 20, clip, width=.51\linewidth]{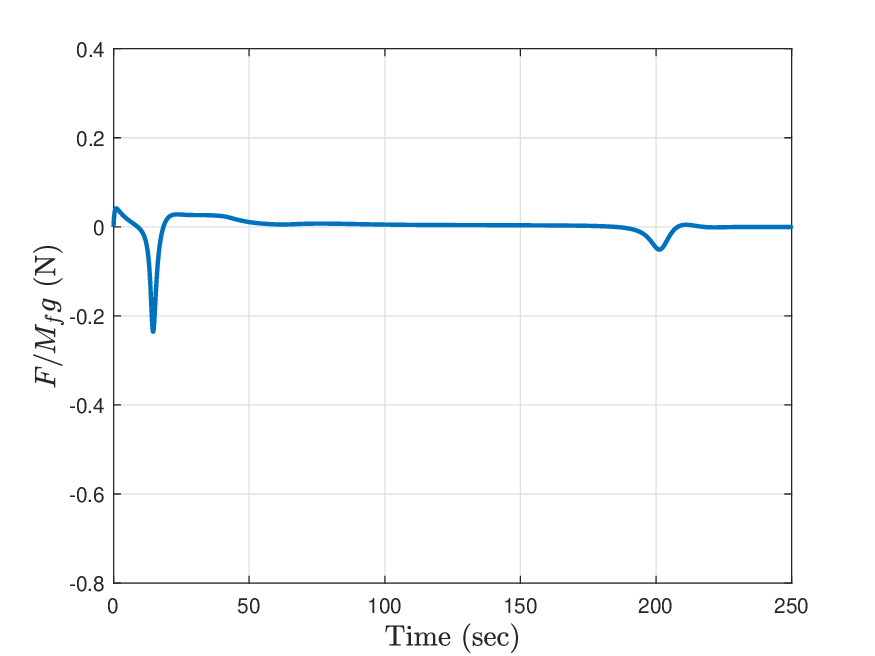}
    \caption{ACC simulation results under a PIDB controller with moderate proportional and derivative control gains for less aggressive and smoother ACC behavior.}
      \label{fig: ACC PIDB - smaller gains}
      \vspace{-4mm}
\end{figure}

We define our PIDB controller $u=-K\bar{x} = -(K_\text{P} (x_2-v_d) + K_\text{I} x_4 + K_\text{D} x_5 + K_\text{B} x_6)$. Selecting the feedback gains to be $K = K_1 = [K_\text{P} \ K_\text{I} \ K_\text{D} \ K_\text{B}] = [50000 \ 5 \ 50000 \ 50000]$ assigns the eigenvalues of the controlled states of the PIDB system to $-29.2765$, $-2.1121$, $-1.0349$ and $-0.0004$. \autoref{fig: ACC PIDB - large gains} shows an effective ACC performance in which the controlled vehicle cruises at the desired speed when possible, i.e. there is a safe time-headway, and slows down when to keep a safe distance when it is not. The vehicle cruises at a speed similar to the leading vehicle to maintain the safe distance when it is very close to it.

However, this controller might be a little aggressive. Tuning the PIDB gains to provide a less aggressive behavior is of course possible, e.g. by selecting the feedback gains to be $K = K_2 = [K_\text{P} \ K_\text{I} \ K_\text{D} \ K_\text{B}] = [1000 \ 5  \ 5000  \ 50000]$ which results in the closed-loop eigenvalues $-2.7977$, $-2.133$, $-0.1987$ and $-0.0217$. We kept the BaS feedback gain high as it helps in making the closed-loop system more conservative and less aggressive in satisfying the desired velocity, i.e. less control input to be used for example for comfort, as shown in \autoref{fig: ACC PIDB - smaller gains}.

\vspace{-2mm}
\subsection{Input-to-State Safely Stable Adaptive Cruise Control}
\begin{figure}[t]
    \centering
    \hspace{-.45cm}
    \includegraphics[trim=15 0 32 20, clip, width=.51\linewidth]{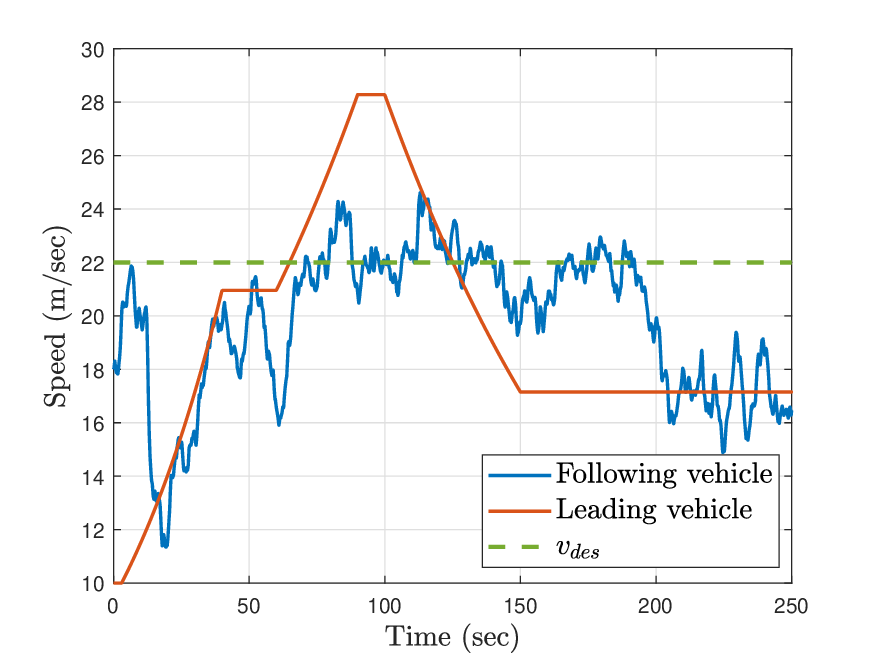}
    \includegraphics[trim=17 0 32 20, clip, width=.51\linewidth]{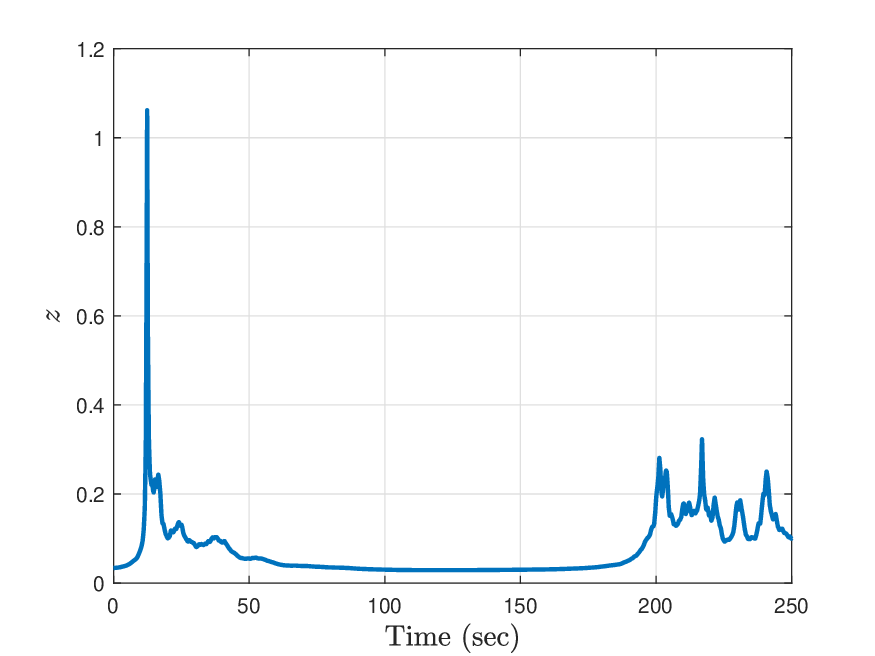}
    \\
    \hspace{-.45cm}
    \includegraphics[trim=15 0 32 20, clip, width=.51\linewidth]{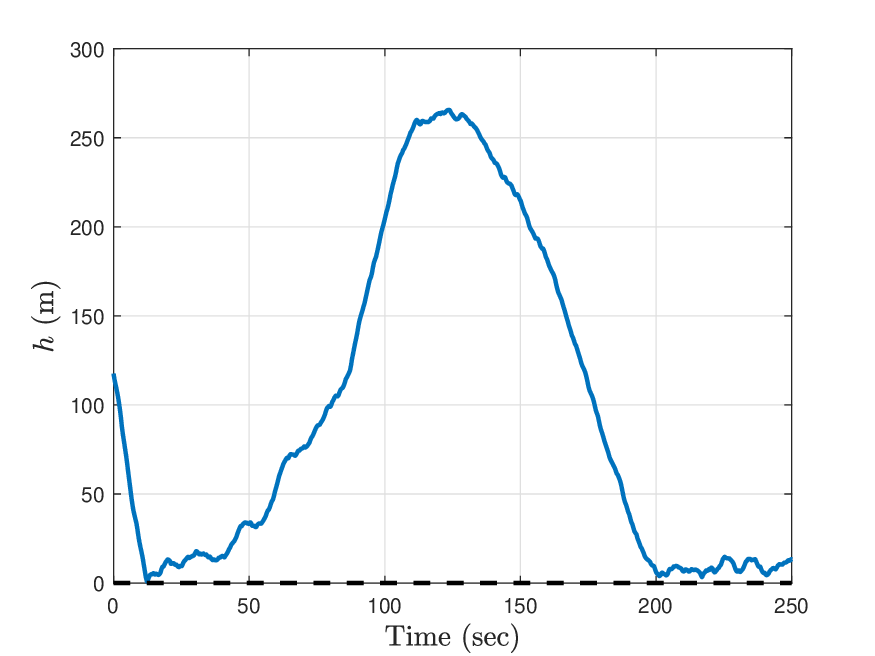}
    \includegraphics[trim=15 0 32 20, clip, width=.51\linewidth]{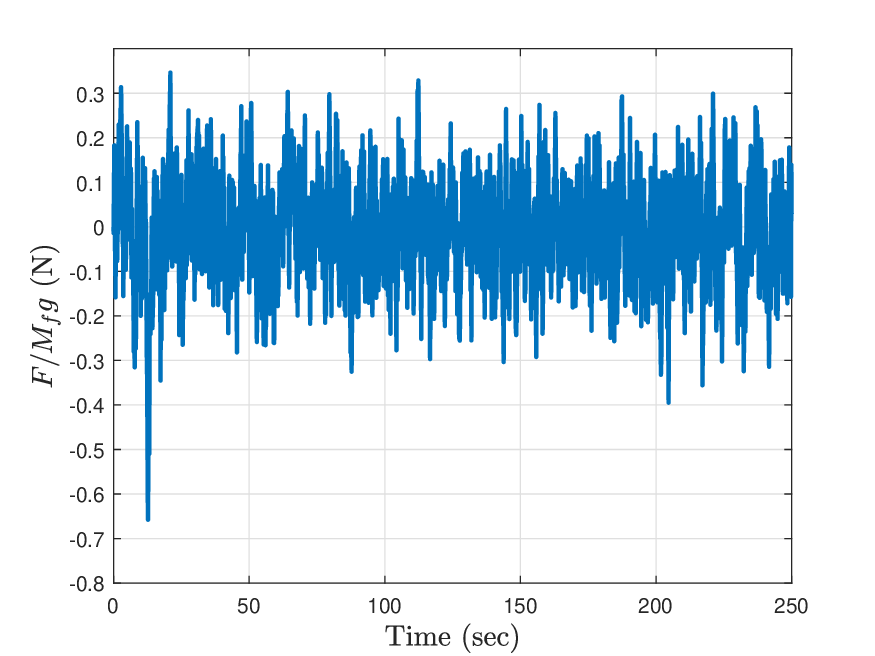}
    \caption{ACC simulation results with input disturbance under a PIDB controller that renders the system IS$^3$. As shown, the BaS value gets very high when the noise pushed the system towards the unsafe limit which is fed-back to the controller to counter act the disturbance for safety.}
      \label{fig: ACC PIDB - large noise - smaller gains}
      \vspace{-5mm}
\end{figure}

As a direct result of the proposed technique in \autoref{sec: input-to-state safe stability}, thanks to the feedback controller of the safety embedded system, the PIDB controller in the preceding example, with the feedback gain vector $K_2$, renders the system LIS$^3$ (\autoref{corollary: linear system AS then IS3}). We rerun the numerical simulations in the foregoing example % but with input disturbance, namely under two scenarios. The first is with some moderate, eventually small disturbance whose amplitude decrease exponentially with the simulation time and the other with a relatively large bounded input disturbance. 
with a relatively large bounded input disturbance, $\frac{||d||_\infty}{M_f g} = 10$. \autoref{fig: ACC PIDB - large noise - smaller gains} shows the numerical simulation in which it can be seen that the safety embedded feedback controller, PIDB, handles the input disturbance well and seeks to satisfy the performance objective when possible while ensuring safety.

\subsection{Mobile Robots Collision Avoidance} 
For this application example, we extend the tutorial example in \autoref{sec: Barrier States}. We consider two simple mobile robots $i$ and $j$ to navigate their way toward prespecified targets while avoiding collision with each other and an obstacle. %The robots, robot $i$ and robot $j$, are represented with 2D single integrators as
%\begin{equation*} \begin{split}
%    \dot{x}_i = \begin{bmatrix} u_{i1} \\ u_{i2}\end{bmatrix}, \dot{x}_j= \begin{bmatrix} u_{j1} \\ u_{j2}\end{bmatrix}
%\end{split}  \end{equation*}
To avoid collision, a BaS is featured through a barrier function that prevents the robots from getting too close to each other. The maximum allowed distance between the two robots is $\delta=0.1$. Therefore, the associated safe set is $\{x_i,x_j \in \mathbb{R}^{2} \ | \ ||x_{i}-x_{j}||^2 > \delta^2 \}$. An obstacle is considered that is represented by a circle at $(0,0)$ with a radius of $0.25$. This calls for a BaS for each agent with respect to the obstacle. Therefore, the overall safe set for each agent is given by $\mathcal{S}=\{x_k\in \mathbb{R}^{2} \ | \ ||x_{i}-x_{j}||^2 > \delta^2 \; \text{and} \; x_{k1}^2 + x_{k2}^2>0.25^2, \;  k=i,j \}$. The Log barrier is selected, $\beta_l=\log(\frac{1+h_l}{h_l})$ for $l=1,2,3$, to construct three barrier states, $z_1, z_2, z_3$, where the first represents the distance constraint between the two robots and the second and third barrier states represent the barriers needed to avoid the obstacle for agent $i$ and agent $j$ respectively. To show the versatility of the proposed method, we design a linear quadratic regulator (LQR) with a quadratic cost function of the form $ \int_{0}^{\infty} \bar{x}^{\top} Q \bar{x} + u^{\top} R u$, with $Q=I_{7\times 7}$ and $R=I_{4\times 4}$, where $\bar{x}=[x_i,x_j,z_1,z_2,z_3]^{\top}$ is the embedded state and $I$ is the identity matrix. As shown in Fig. \ref{fig: Robots with obstacle}, using the proposed technique, we are effectively able to send the two robots safely to the targeted positions while avoiding the obstacle as well as colliding with each other. The robots get very close to each other at sometime but never get too close, i.e. the distance between them is greater than $\delta$. 

 \begin{figure}
\vspace{-2mm}
\centering
     \subfloat{\includegraphics[trim=0 0 0 0, clip, width=0.8\linewidth]{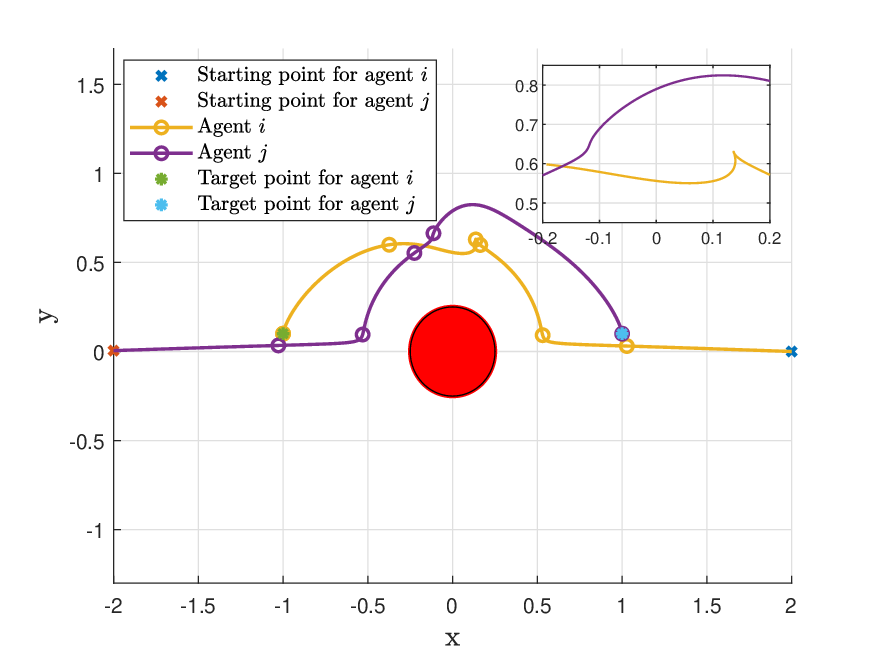}}
     \vspace{-4mm}
    \subfloat{\includegraphics[trim=0 5 0 0, clip, width=0.8\linewidth]{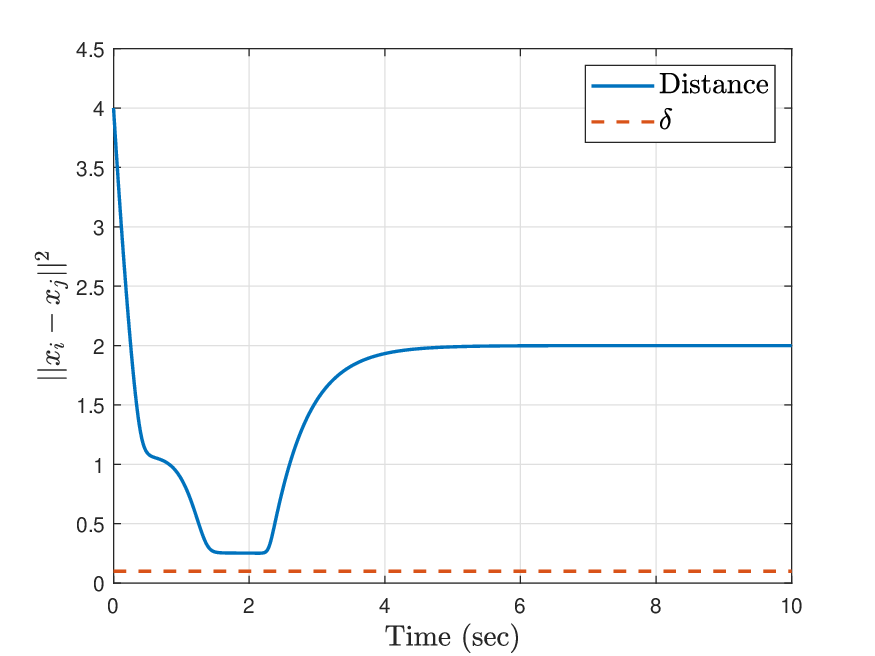}} \caption{Two mobile robots, shown in yellow and purple, starting from two opposite initial positions sent to two opposite final positions with an obstacle in-between. The robots switch positions while avoiding collision with each other and the obstacle. It can be seen that the two robots start by moving in a straight line toward their target positions but start to steer away to avoid the obstacle and then take some turns to avoid crossing the same point at the same time and to avoid getting too close (less than $\delta$ distance) to each other as shown in the bottom figure.}
      \label{fig: Robots with obstacle}
      \vspace{-5mm}
   \end{figure} 

 \vspace{-4mm}
\section{Conclusion} \label{sec: Conclusion}
This paper formally developed and generalized the concept of barrier states (BaS) and safety embedded systems for safe multi-objective control. By embedding safety objectives as system states, the proposed approach transforms constrained control problems into unconstrained ones, enabling the simultaneous achievement of safety and performance objectives when possible. The key advantages of this framework include:
\begin{itemize}
    \item \textbf{Intrinsic Barrier Rate:} The barrier’s rate of change is determined by the control law rather than manual tuning, ensuring a naturally regulated safety embedded system.
    \item \textbf{Flexibility with Control Methods:} The approach is compatible with various control techniques, including linear, optimal, model predictive (MPC), adaptive, and robust control \cite{oshin2024diffRobMPC,cho2023mpcbas,Almubarak2021SafetyEC,almubarak2021safeddp,almubarak2021safeminmax,aoun2023l1}.
    \item \textbf{Relative Degree Agnostic:} Unlike many existing methods, this framework does not impose specific relative degree requirements, adding flexibility to safe control synthesis.
    \item \textbf{Unified Control Objectives:} The integration of safety and performance constraints into a single unconstrained control problem allows for multiple safety constraints to be addressed through either a single or multiple barrier states, simplifying control design and analysis.
\end{itemize}
While the proposed method offers substantial benefits, it naturally inherits some well-known challenges in nonlinear control. Augmenting the system increases its dimensionality and introduces nonlinearity, necessitating nonlinear control techniques for analysis and synthesis. Additionally, the achievable region of attraction (ROA) depends on the chosen control design, and computational complexity may vary depending on the selected control method. However, these factors represent practical trade-offs rather than inherent limitations of the framework.

The method was extended to handle input constraints and robustness against external noise, leading to the development of input-to-state safety (ISSf) via barrier states and the novel concept of input-to-state safe stability (IS$^3$). Numerical simulations in safety-critical applications validated its effectiveness, including its ability to prioritize safety when performance objectives cannot be simultaneously achieved, as demonstrated by the PIDB controller for adaptive cruise control.

\printbibliography
   
\begin{IEEEbiography}[{\includegraphics[width=1in,height=1.25in,clip,keepaspectratio]{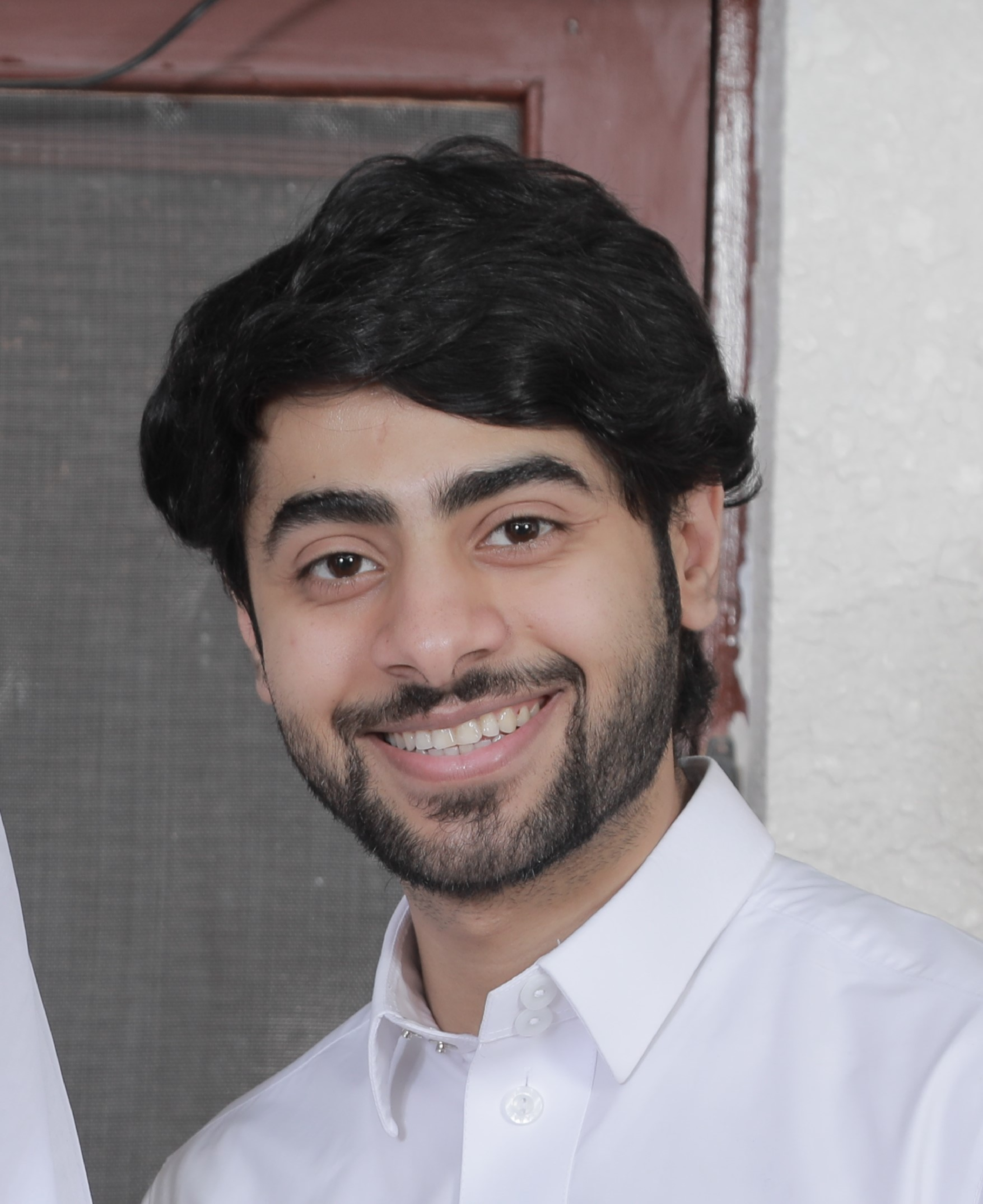}}] {Dr. Hassan Almubarak} received his Ph.D. from the Autonomous Control and Decision Systems Laboratory at Georgia Institute of Technology in 2024. He earned his M.S. in Electrical and Computer Engineering from Georgia Tech in 2018 and his B.Sc. in Control and Instrumentation Systems Engineering from King Fahd University of Petroleum and Minerals (KFUPM), Dhahran, Saudi Arabia, in 2016. He is currently an Assistant Professor in the Department of Control and Instrumentation Engineering at KFUPM. Prior to joining KFUPM, he was a Research Engineer at General Electric (GE) Vernova Advanced Research Center in Niskayuna, New York, USA, where he architected generative AI-driven solutions utilizing language models and cloud technologies to optimize complex processes. He also worked on developing various advanced model predictive control technologies for intelligent multi-agent energy systems. His research interests include multi-objective and safety-critical control, optimal control, dynamic optimization, machine learning, and AI, with applications in robotics, autonomous systems, and intelligent systems optimization.

\end{IEEEbiography}

\begin{IEEEbiography}[{\includegraphics[width=1in,height=1.25in,clip,keepaspectratio]{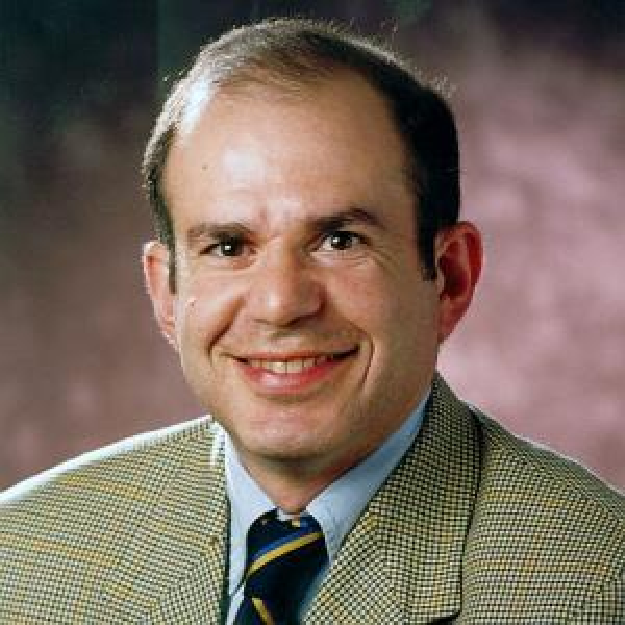}}]{Dr. Nader Sadegh} has been on the faculty of the Woodruff School of Mechanical Engineering at the Georgia Institute of Technology since January of 1988. He received his Ph.D. and M.S. degrees from the University of California at Berkeley in 1987 and 1984, respectively. His early research work in the field of robotics and automation resulted in the development of a class of adaptive and learning controllers for nonlinear mechanical systems including robotic manipulators. Parallel to his efforts in the repetitive learning control area, Dr. Sadegh has made significant contributions to the theory and applications of artificial neural networks (ANN). In particular, he developed a new class of networks that enable a machine to learn the ‘concept’ of executing a series of similar tasks rather than a specific one as is done in repetitive control through practice. Dr. Sadegh’s most recent and ongoing work are in the areas of robotic motion planning, infinite dimensional boundary control systems, and nonlinear optimal control. Dr. Sadegh has authored and co-authored over 150 archival publications and holds a joint patent for the application of his learning controller to a robotic system. He has been the director of Georgia Tech’s interdisciplinary Robotics graduate program since July 2017. He has also served as an associate editor for the ASME journal of Dynamic Systems, Measurement, and Control, and is currently on the DSCC Conference Editorial Board.
\end{IEEEbiography}

\begin{IEEEbiography}[{\includegraphics[width=1in,height=1.25in,clip,keepaspectratio]{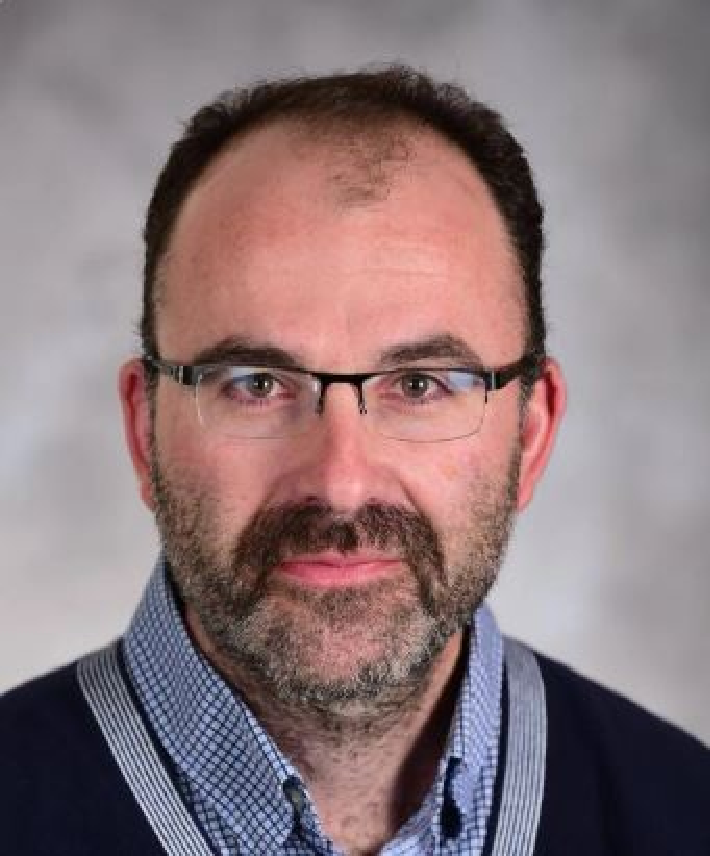}}]{Dr. Evangelos A. Theodorou} is an Associate Professor with the Daniel Guggenheim School of Aerospace Engineering at Georgia Institute of Technology. He is also the director of the Autonomous Control and Decision Systems Laboratory, and he is affiliated with the Institute of Robotics and Intelligent Machines and the Center for Machine Learning Research at Georgia Tech. Dr. Theodorou holds a BS in Electrical Engineering,  from the Technical University of Crete (TUC), Greece in 2001 and three MSc degrees in Production Engineering from TUC in 2003, Computer Science and Engineering from University of Minnesota in 2007, and Electrical Engineering from the University of Southern California (USC) in 2010. In 2011, he graduated with his PhD in Computer Science from USC. From 2011 to 2013, he was a Postdoctoral Research Fellow with the department of Computer Science and Engineering, University of Washington. Dr. Theodorou is the recipient of the King-Sun Fu best paper award of the IEEE Transactions on Robotics in 2012 and recipient of several best paper awards and nominations in machine learning and robotics conferences. His theoretical research spans the areas of stochastic optimal control theory, machine learning, statistical physics and optimization. 
\end{IEEEbiography}

\end{document}